\documentclass[journal]{IEEEtran}
\usepackage{amsmath,amsfonts,amssymb}
\usepackage{mathtools}
\usepackage{array}
\usepackage{textcomp}
\usepackage{stfloats}
\usepackage{url}
\usepackage{verbatim}
\usepackage{graphicx}
\usepackage{cite}
\usepackage{threeparttable}
\usepackage{algorithm} 
\usepackage{algorithmicx}
\usepackage{subfigure}
\usepackage[noend]{algpseudocode}
\usepackage{color}
\usepackage{mathrsfs}
\usepackage{hyperref}
\usepackage{multicol}
\usepackage{multirow}
\usepackage{wasysym}
\usepackage{makecell}
\usepackage{stmaryrd}
\usepackage{dashrule}
\usepackage{enumitem}
\usepackage{cite}

\usepackage[normalem]{ulem}
\newtheorem{definition}{Definition}
\newtheorem{theorem}{Theorem}  
\newtheorem{remark}{Remark}
\newtheorem{lemma}{Lemma}
\newtheorem{proof}{Proof}

\begin{document}

\title{Bandwidth-Efficient Two-Server ORAMs with $O(1)$ Client Storage}

\author{Wei Wang,
        Xianglong Zhang,
        Peng Xu,
        Rongmao Chen,
        Laurence Tianruo Yang

\thanks{Wei Wang and Xianglong Zhang are with the School of Computer Science and
 Technology, Huazhong University of Science and Technology, Wuhan 430074,
 China (e-mail: viviawangwei@hust.edu.cn; vrwudi@gmail.com).}
\thanks{Peng Xu is with the Hubei Key Laboratory of Distributed System Security, School of Cyber Science and Engineering, Huazhong University of Science and Technology, Wuhan 430074, China, and also with the Jinyinhu Laboratory, Wuhan 430040, China (e-mail: xupeng@mail.hust.edu.cn).}
\thanks{Rongmao Chen is with the School of Computer, National University of Defense Technology, Chang sha 410073, China (e-mail: chromao@nudt.edu.cn).}
\thanks{ Laurence T. Yang is with the School of Computer Science and Technology,
 Huazhong University of Science and Technology, Wuhan 430074, China,
 also with the School of Computer and Artificial Intelligence, Zhengzhou
 University, Zhengzhou 450001, China, and also with the Department of
 Computer Science, St. Francis Xavier University, Antigonish, NS B2G 2W5,
 Canada (e-mail: ltyang@ieee.org).}
}

\maketitle

\begin{abstract}
Oblivious RAM (ORAM) allows a client to securely retrieve elements from outsourced servers without leakage about the accessed elements or their virtual addresses. Two-server ORAM, designed for secure two-party RAM computation, stores data across two non-colluding servers. However, many two-server ORAM schemes suffer from excessive local storage or high bandwidth costs. To serve lightweight clients, it is crucial for ORAM to achieve concretely efficient bandwidth while maintaining $O(1)$ local storage. Hence, this paper presents two new client-friendly two-server ORAM schemes that achieve practical logarithmic bandwidth under $O(1)$ local storage, while incurring linear symmetric key computations. The core design features a hierarchical structure and a pairwise-area setting for the elements and their tags. Accordingly, we specify efficient read-only and write-only private information retrieval (PIR) algorithms in our schemes to ensure obliviousness in accessing two areas respectively, so as to avoid the necessity of costly shuffle techniques in previous works. We empirically evaluate our schemes against LO13 (TCC'13), AFN17 (PKC'17), and KM19 (PKC'19) in terms of both bandwidth and time cost. The results demonstrate that our schemes reduce bandwidth by approximately $2\sim{4}\times$ compared to LO13, and by $16\sim{64}\times$ compared to AFN17 and KM19. For a database of size $2^{14}$ blocks, our schemes are over $64\times$ faster than KM19, while achieving similar performance to LO13 and AFN17 in the WAN setting, with a latency of around $1$ second.
\end{abstract}


\section{Introduction}
The development of outsourced storage has encouraged lightweight clients to store private data on encrypted databases, such as Always Encrypted \cite{AWE} and MongoDB \cite{MongoDB}. This has raised concerns about access privacy beyond data security, as these systems reveal access patterns when clients retrieve data from the database. Studies \cite{DBLP:conf/ndss/IslamKK12, DBLP:conf/ccs/CashGPR15, DBLP:conf/ccs/KellarisKNO16} show that such leakage leads to inference attacks and compromises database security.

\begin{table*}[!t]
{
    \begin{center}
    \caption{Comparisons of two-server ORAM schemes on client storage, communication cost, bandwidth, and computations.}
    \label{CT} 
    \begin{threeparttable}
    \setlength{\tabcolsep}{3.2mm}
    {
    \begin{tabular}{|l|c|c|c|c|c|c|c|}
        \hline
        \textbf{Scheme} & \textbf{\makecell[c]{Client Storage\\ (\# block)}} & \textbf{\makecell[c]{Communication\\ (bit)}} & \textbf{\makecell[c]{Bandwidth\tnote{$\dagger$}\\ (\# block)}} & \textbf{\makecell[c]{Block size\\ (bit)}} & \textbf{\makecell[c]{Computation\\ (server)}} & \textbf{\makecell[c]{Computation\\ (client)}}
        \\
        \hline
        \hline
        GKW18 \cite{DBLP:conf/asiacrypt/GordonKW18} & $O(\log{N})$ & $O(B\log{N})$ & $\sim{10}\log{N}$  & $\Omega(\log{N})$ & linear & polylog\\
        \hline
        \hline
        LO13 \cite{DBLP:conf/tcc/LuO13} & $O(1)$ & $O(B\log{N})$ & $\sim{160}\log{N}$  & $\Omega(\log{N})$ & polylog & polylog\\
        \hline
        AFN17 \cite{DBLP:conf/pkc/AbrahamFNP017} & $O(1)$ & $O(\frac{B\log{N}+\log^{4}{N}}{\log\log{N}})$ & $O(\frac{\log{N}}{{\log\log{N}}})$  & $\Omega(\log^3{N})$ & polylog & polylog\\
        \hline
        KM19\tnote{$\ddagger$}\ \  \cite{DBLP:conf/pkc/KushilevitzM19}  & $O(1)$ & $O(\frac{B\log{N}+\log^{3}{N}}{\log\log{N}})$ & $O(\frac{\log{N}}{{\log\log{N}}})$  & $\Omega(\log^2{N})$ & polylog & polylog\\
        \hline
        \hline
        Cforam\tnote{$\S$} & $O(1)$ & $O(B\log{N}+\kappa\log^{2}{N})$ & $\sim{24}\log{N}$ & $\Omega(\kappa\log{N})$ & linear & polylog\\
        \hline
        Cforam+ & $O(1)$ & $O(B\log{N})$ & $\sim{34}\log{N}$ & $\Omega(\log{N})$ & linear & polylog\\
        \hline
    \end{tabular}
               
    \begin{tablenotes}
        \item[$\dagger$] The bandwidth consumption assumes blocks of at least the specified size. Thus, a scheme is considered more efficient and scalable if it achieves the claimed bandwidth with a smaller block size. For example, given an array with $N = 2^{20}$ blocks, each block of size $32$ bytes, AFN17 and KM19 consume over $512$ KB per access, while our schemes' consumption does not exceed $16$ KB.

        \item[$\ddagger$] To realize two-server ORAM, KM19 \cite{DBLP:conf/pkc/KushilevitzM19} needs to construct an oblivious hashing based on oblivious sort. Using Zig-zag sort \cite{DBLP:conf/stoc/Goodrich14} can achieve sub-logarithmic bandwidth, but with a large constant. Leveraging Bitonic sort \cite{DBLP:conf/afips/Batcher68} increases both the bandwidth and roundtrips to $\omega(\log{N})$.

        \item[$\S$] $\kappa$ represents the DPF key length.
        
    \end{tablenotes}
    }
    \end{threeparttable} 
    \end{center}
}
\end{table*}

To address the issue above, Oblivious RAM (ORAM), originally proposed by Goldreich and Ostrovsky \cite{DBLP:journals/jacm/GoldreichO96}, has been widely regarded as an effective countermeasure since it can fully hide the access patterns. Specifically, ORAM enables a client to outsource his/her database to untrusted servers while performing read/write operations obliviously such that the access patterns, i.e., accessed memory locations, leak no information about the operations. Over the years, ORAM has attracted widespread attention and spawned various branches of study, such as lower bound \cite{DBLP:conf/tcc/WeissW18,DBLP:conf/crypto/LarsenN18,DBLP:conf/innovations/BoyleN16,DBLP:conf/crypto/KomargodskiL21}, optimal bandwidth \cite{DBLP:conf/soda/KushilevitzLO12,DBLP:conf/ccs/StefanovDSFRYD13,DBLP:conf/focs/PatelP0Y18,DBLP:conf/eurocrypt/AsharovKLNPS20,DBLP:conf/crypto/AsharovKLS21}, differential obliviousness \cite{DBLP:conf/soda/ChanCMS19,DBLP:conf/approx/BeimelNZ19,DBLP:conf/acns/GordonKLX22,DBLP:conf/eurocrypt/ZhouSCM23}, locality preserving \cite{DBLP:conf/eurocrypt/AsharovCNP0S19,DBLP:conf/ndss/ChakrabortiACMR19}, malicious security \cite{DBLP:journals/jacm/GoldreichO96,DBLP:conf/ccs/StefanovDSFRYD13,DBLP:conf/hpec/RenFYDD13,DBLP:conf/crypto/MathialaganV23}, distributed settings \cite{DBLP:conf/tcc/LuO13,DBLP:conf/pkc/KushilevitzM19,DBLP:conf/asiacrypt/GordonKW18,DBLP:conf/pkc/HamlinV21,DBLP:conf/ccs/DoernerS17,DBLP:conf/sp/ZahurW0GDEK16,DBLP:conf/uss/VadapalliHG23}, and more.

However, ORAM has long been regarded as a costly approach to prevent access pattern leakage. It imposes significant overheads, either in bandwidth\footnote{In ORAM, \textit{bandwidth} refers to the amount of data blocks transferred between the client and the servers during an access operation.} \cite{DBLP:journals/jacm/GoldreichO96} or local storage \cite{DBLP:conf/ccs/StefanovDSFRYD13}, making it impractical for resource-constrained clients, including those using trusted execution environments (TEEs) \cite{DBLP:conf/ndss/SasyGF18} or wireless devices. To address these challenges, researchers shifted their focus to improving efficiency, leading to the development of hierarchical ORAM \cite{DBLP:journals/jacm/GoldreichO96} and numerous optimizations \cite{DBLP:conf/soda/KushilevitzLO12,DBLP:conf/focs/PatelP0Y18,DBLP:conf/eurocrypt/AsharovKLNPS20,DBLP:conf/crypto/AsharovKLS21}. Proposed by Asharov et al., OptORAMa \cite{DBLP:conf/eurocrypt/AsharovKLNPS20} achieves $O(\log{N})$ bandwidth and $O(1)$ local storage for a database of size $N$. Yet, like most hierarchical ORAMs \cite{DBLP:journals/jacm/GoldreichO96,DBLP:conf/asiacrypt/ChanGLS17,DBLP:conf/focs/PatelP0Y18}, the bandwidth has a large constant in front of the $O$ ($\sim{10}^{4}$ \cite{DBLP:conf/scn/DittmerO20}), limiting its practicality.

To achieve more efficient oblivious access, the ORAM deployment across multiple servers has been proposed. Following the definition in \cite{DBLP:conf/pkc/KushilevitzM19}, we represent a multi-server ORAM with a parameter tuple $(m,t)$, where $m > 1$ denotes the total number of servers, and $t < m$ specifies the maximum number of colluding servers tolerated. In this paper, we focus on the two-server setting\footnote{This paper considers the standard two-server scenario, where the client interacts with two servers independently. A similar concept, distributed ORAM, involves server-to-server communication to process client requests. See details in Appendix \ref{AppDiss}.}. The two-server model strikes a balance between practicality and efficiency \cite{DBLP:conf/tcc/LuO13,DBLP:conf/asiacrypt/GordonKW18,DBLP:conf/pkc/AbrahamFNP017,DBLP:conf/pkc/KushilevitzM19}, offering significantly better performance than single-server solutions \cite{DBLP:conf/ccs/StefanovDSFRYD13,DBLP:conf/eurocrypt/AsharovKLNPS20}, while incurring lower deployment costs than configurations with $m \geq 3$ servers \cite{DBLP:conf/asiacrypt/ChanKNPS18,DBLP:conf/pkc/KushilevitzM19,DBLP:conf/acns/JareckiW18,DBLP:conf/uss/VadapalliHG23}. Moreover, the two-server setting is well-suited for two-party RAM computation scenarios \cite{OS97,DBLP:conf/tcc/LuO13,DBLP:conf/ccs/DoernerS17,DBLP:conf/pkc/HamlinV21,DBLP:conf/uss/VadapalliHG23}.

Most existing two-server ORAM schemes adopt shuffle-based strategies to achieve obliviousness. For example, Lu and Ostrovsky \cite{DBLP:conf/tcc/LuO13} developed the first $(2,1)$-ORAM scheme that transforms all elements in hash tables between the client and servers for shuffling, consuming approximately $160\log{N}$ bandwidth.
Kushilevitz and Mour \cite{DBLP:conf/pkc/KushilevitzM19} introduced a three-server ORAM scheme based on a balanced hierarchical structure and extended it to the two-server setting by oblivious hashing, which incurs significant bandwidth and roundtrip costs.
Abraham et al. \cite{DBLP:conf/pkc/AbrahamFNP017} constructed a tree-based ORAM instance with sub-logarithmic bandwidth. The instance leverages a $d$-ary tree to store elements and requires a recursive ORAM to store a position map. It cannot achieve practical bandwidth for small block sizes, e.g., $o(\log^{3}{N})$. In contrast, a computation-based ORAM scheme introduced by Gordon et al. \cite{DBLP:conf/asiacrypt/GordonKW18}, achieves low bandwidth and linear computation overheads. However, similar to path ORAM \cite{DBLP:conf/ccs/StefanovDSFRYD13}, it requires $O(\log{N})$ local storage.

Given the current landscape of two-server ORAMs and the needs of lightweight clients, we pose the following question:

\textit{Is it possible to construct a client-friendly two-server ORAM scheme that  achieves $O(1)$ local storage and concretely efficient bandwidth for small data blocks?}

\subsection{Contributions}

In this paper, we answer the above question affirmatively by introducing two schemes: Cforam and Cforam+. A comparison of our schemes with related works is provided in Table \ref{CT}.

We introduce a two-server ORAM model and construct Cforam and Cforam+ to achieve our goals. At a high level, Cforam adopts a hierarchical structure, which helps maintain $O(1)$ local storage. To reduce the high bandwidth overhead caused by shuffling, it employs a pairwise-area setting to store elements with tags and distributed point function (DPF)-based PIR for efficient access. Additionally, by exploring cyclic shift optimization of PIR, we derive Cforam+, which achieves practical $O(\log{N})$ bandwidth for small logarithmic block sizes. We implement our schemes and compare them with LO13 \cite{DBLP:conf/tcc/LuO13}, AFN17 \cite{DBLP:conf/pkc/AbrahamFNP017}, and KM19 \cite{DBLP:conf/pkc/KushilevitzM19} under various settings. Cforam and Cforam+ reduce bandwidth by $2\sim4\times$ compared to LO13 and by $16\sim64\times$ compared to AFN17 and KM19. In terms of time cost, while our schemes incur moderately higher latency than LO13 in the LAN setting due to linear symmetric key computations, they maintain comparable performance to LO13 and AFN17 in the WAN setting and are more than $16\times$ faster than KM19.

We examine a realistic tradeoff between bandwidth and computational resources. For lightweight clients, such as mobile phones, affordable and secure outsourced services are often more viable than enhancing their own capabilities.
For instance, consider two Amazon EC2 servers\footnote{\url{https://aws.amazon.com/ec2/pricing/on-demand/}} priced at $\$0.023$/CPU-hour, alongside a client using Mint Mobile\footnote{\url{https://www.mintmobile.com/plans/}} with a plan costing $\$15$ for $5$GB data per month, equating to approximately $m\$3$/MB on average\footnote{Here, $m\$$ denotes milli-dollars, where $m\$1 = \$0.001$.}. In the WAN setting with a database size of $N=2^{14}$, the cost per access is $m\$0.07$ for LO13, and exceeds $m\$1$ for AFN17 and KM19. In contrast, our schemes reduce the cost to about $m\$0.03$, making them more affordable and practical than others.

\subsection{Technical Overview}\label{SecDG}

We outline the essential concepts and techniques employed: a hierarchical structure and symmetric encryption for element storage, two-server PIR to obfuscate access patterns while avoiding costly oblivious hashing, and a pairwise-area design and cyclic shift optimization for efficient overhead.

\smallskip

\noindent\textbf{Hierarchical structure.} Our starting point is a hierarchical structure consisting of $1+\log{N}$ levels, each associated with a hash table, denoted as $T_{0},...,T_{\log{N}}$, for managing $N$ elements. Each table $T_{i}$ can store up to $2^{i}$ elements. This structure is employed in hierarchical ORAM, which includes three phases: Setup, Access, and Rebuild, as detailed in Appendix \ref{AppHORAM}. To achieve oblivious access in this structure, two key requirements must be satisfied: (i) indistinguishability of elements; (ii) indistinguishability of memory access locations.

\smallskip

\noindent\textbf{Encrypted  elements.} 
To satisfy requirement (i), all elements need to be encrypted with symmetric encryption schemes, such as AES, before being stored on the servers. Additionally, the client must re-encrypt the elements before sending them to the servers during the access and rebuild phases. The encryption and re-encryption processes are essential and commonly applied as fundamental procedures in most ORAMs \cite{DBLP:journals/jacm/GoldreichO96,DBLP:conf/tcc/LuO13,DBLP:conf/pkc/KushilevitzM19}.

\smallskip

\noindent\textbf{Oblivious hashing vs. Two-server PIR.} To meet requirement (ii), the client must access seemingly random locations within each level for each query. Otherwise, an untrusted server may trace accessed locations and infer client requests.

Oblivious hashing is a viable solution for obfuscating accessed locations. However, most oblivious hashing schemes \cite{DBLP:conf/asiacrypt/ChanGLS17,DBLP:conf/icalp/GoodrichM11,DBLP:conf/soda/KushilevitzLO12} rely on costly oblivious sorting techniques. These methods include bitonic sort \cite{DBLP:conf/afips/Batcher68}, which requires $O(n\log^{2}{n})$ bandwidth, and AKS sort \cite{DBLP:conf/stoc/AjtaiKS83} and Zig-zag sort \cite{DBLP:conf/stoc/Goodrich14}, which achieve $O(n\log{n})$ bandwidth but with a large constant factor. Other schemes that avoid oblivious sorting still rely on expensive random shuffle algorithms \cite{DBLP:conf/focs/PatelP0Y18,DBLP:conf/eurocrypt/AsharovKLNPS20,DBLP:conf/crypto/AsharovKLS21}, causing performance bottlenecks for hierarchical ORAM.

To overcome the limitations of oblivious hashing, a practical alternative is to leverage two-server read-only PIR \cite{DBLP:conf/pkc/KushilevitzM19} for each level, preventing the servers from learning the client’s requests. With this approach, we can use a standard hash table to store elements, avoiding costly oblivious hashing. In Appendix \ref{AppPIR}, we show how to construct a read-only PIR scheme using the distributed point function (DPF) \cite{DBLP:conf/eurocrypt/BoyleGI15,DBLP:conf/ccs/BoyleGI16}, while achieving low bandwidth overhead.

\smallskip

\noindent\textbf{Elements deduplication.} Although two-server PIR can effectively hide access patterns at each level, it alone cannot construct a complete ORAM. A fundamental step in hierarchical ORAM is that each accessed element must be written back to the first level of the hierarchical structure, without being deleted from its original level. This results in duplicated copies with the same virtual addresses across different levels. Therefore, deduplication is crucial. Two main approaches exist: (i) Oblivious sorting \cite{DBLP:journals/jacm/GoldreichO96,boyle2015oblivious} during the rebuild, which organizes elements by their virtual addresses and sequentially removes duplicates. This method incurs significant bandwidth overhead, making it impractical. (ii) Overwriting the accessed position with a dummy element during the access \cite{DBLP:conf/tcc/LuO13,DBLP:conf/asiacrypt/GordonKW18,DBLP:conf/eurocrypt/AsharovKLNPS20}, which avoids sorting but reveals read/write positions. Asharov et al. \cite{DBLP:conf/crypto/AsharovKLS21} proposed a linear-time oblivious deduplication scheme, assuming that the input array is randomly shuffled. Our design relies on the second method. Since this method exposes write positions, we need to hide them to ensure private deduplication.

\smallskip

\noindent\textbf{Pairwise-area setting.} In the aforementioned context, PIR-write \cite{DBLP:conf/pkc/KushilevitzM19} is a promising strategy for privately writing a value to a covert position. Unfortunately, there is currently no efficient two-server PIR that supports simultaneous read and write operations. Inspired by the $tagging$ technique of Lu and Ostrovsky \cite{DBLP:conf/tcc/LuO13}, we devised a pairwise-area setting for data storage in Cforam, which consists of an element area and a tag area. Both areas follow a hierarchical structure. The element area uses read-only PIR to access elements, while the tag area employs write-only PIR to mark the status of the corresponding elements. This design efficiently solves the deduplication problem while ensuring security.

\smallskip

\noindent\textbf{Cyclic shift optimization.} The aforementioned approach requires the client to generate $O(\log{N})$ DPF keys, each with a size of $O(\kappa\log{N})$, where $\kappa$ denotes the size of a DPF key. For a small block size $B=\Omega(\log{N})$, this results in a bandwidth overhead of $\omega(\log{N})$. To further reduce the bandwidth, we propose a so-called cyclic shift optimization strategy in Cforam+, which reduces the client's burden to send only $O(1)$ DPF keys along with some cyclic shift information, enabling $O(\log{N})$ PIR operations. This optimization achieves logarithmic bandwidth overhead for $B=\Omega(\log{N})$. A detailed construction is provided in Sec. \ref{SecOptScheme}.

\section{Building Blocks}\label{SecPrelim}
This section introduces some frequently used notations and presents foundational building blocks of our schemes, such as cuckoo hashing and two-server PIR.

\subsection{Notations}

\begin{table}[!t]
    \center
    \caption{Notations.} 
    \label{TabNot}
    \begin{threeparttable}
    \setlength{\tabcolsep}{1.8mm}
    {
    \begin{tabular}{|c|l|}
        \hline
        Notations & Descriptions \\
        \hline
        {$B$} & A block (or element) size. \\
        \hline
        {$\Upsilon$} & A tag size. \\
        \hline
        {$N$} & Database size. \\
        \hline
        {$\lambda$} & \makecell[l]{Security parameter of cuckoo hashing, which\\ is equal to $N$ as applied in \cite{DBLP:conf/asiacrypt/ChanGLS17}.}\\
        \hline
        {$\kappa$} & Security parameter of DPF. \\
        \hline
        {$\mathcal{S}_{b}$} & One of the two servers for $b\in\{0,1\}$. \\
        \hline
        $\langle{X}\rangle$ & A replication of the array $X$.\\
        \hline
        $\langle{X}\rangle_{b}$ & $\mathcal{S}_{b}$'s replication of $X$.\\
        \hline
        $\llbracket{X}\rrbracket$ & A secret sharing of the array $X$.\\
        \hline
        ${\llbracket{X}\rrbracket}_{b}$ & $\mathcal{S}_{b}$'s share of $X$.\\
        \hline
        $EX_{b}$ & Storage structure of $\mathcal{S}_{b}$.\\
        \hline
        $Out_{b}$ & Obtained elements from $\mathcal{S}_{b}$ for a request.\\
        \hline
        $AP_{b}$ & Access patterns of $\mathcal{S}_{b}$ for a request.\\
        \hline
        $V_{i}$ & \makecell[l]{The $i$th unit vector, i.e., $V_{i}[j]=1$ if $j\equiv{i}$;\\ otherwise, $V_{i}[j]=0$.} \\
        \hline
        $\bot$ & Null value.\\
        \hline
        $EB|TB$ & Element$|$Tag buffer.\\
        \hline
        $ES|TS$ & Element$|$Tag stash.\\
        \hline
        $ET_{i}|TT_{i}$ & \makecell[l]{Element$|$Tag hash table of the $i$th level, where
        \\ $ET_{i}=(ET_{i0},ET_{i1})\ |\ TT_{i}=(TT_{i0},TT_{i1})$,\\ denoting two tables of cuckoo hashing.} \\
        \hline
        {$LenB$} & Buffer length. \\
        \hline
        {$LenS$} & Stash length. \\
        \hline
        {$Len_{i}$} &  Hash table length of the $i$th level. \\
        \hline
    \end{tabular}

    }
    \end{threeparttable}
\end{table}

We denote $X = (X[0], \ldots, X[n-1])$ as an element array with cardinality $|X| = n$. Each element in the array is described as $(a,v)$, where $a$ is the virtual address, $v$ is the value, and $\tau$ is its associated tag. We
represent $\langle{X}\rangle$ as a replication of $X$ and $\llbracket{X}\rrbracket$ as a secret sharing of $X$. Note that the $\llbracket{X}\rrbracket$ can either be additive or XOR shares. In this paper, we use XOR shares of $X$ for consistency. We use $\mathcal{S}_{b}$ to denote one of the two servers where $b\in\{0,1\}$. For server $\mathcal{S}_{b}$, $EX_{b}$ refers to its storage structure, $Out_{b}$ represents the output elements for an access request, and $AP_{b}$ refers to the memory access pattern, i.e., a sequence of accessed memory locations. Additional notations are provided in Table \ref{TabNot}.

\subsection{Cuckoo Hashing}\label{SecCuckoo}
Cuckoo hashing \cite{DBLP:journals/jal/PaghR04} is a powerful primitive that enables storing elements in a small space for efficient querying.
This paper considers two-table cuckoo hashing.

\begin{definition}[Syntax of Two-table Cuckoo Hashing]
\normalfont
    A two-table cuckoo hashing scheme, $\Pi_{CuckHash}$, consists of two algorithms, $(Build, Lookup)$, defined as follows.

    \begin{itemize}[leftmargin=*]
        \item $Build(1^\lambda, X)$: A build algorithm that is given a security parameter $\lambda$ and an element array $X=\{(a_{1},v_{1}),\ldots,(a_{n},v_{n})\}$ as inputs, returns a construction $T$, which includes two hash tables and a stash that successfully allocates $X$, or $\bot$ if the allocation fails. 

        \item $Lookup(a,T)$: A query algorithm that is given a virtual address $a$ and a construction $T$, returns the required element $(a,v)$ if $(a, v)\in{X}$, or $\bot$ otherwise.
    \end{itemize}

    The \textbf{build correctness} guarantees that the failure probability of the build phase is negligible in $\lambda$. In other words, the number of elements written to the stash should not exceed its threshold. Formally, for any element array $X$, the following holds:
    \begin{equation}
        \Pr\left[
        {Build}(1^{\lambda},X)=\bot
        \right]\leq{negl(\lambda)}
    \end{equation}
    
    The \textbf{lookup correctness} requires that the client can always extract the latest value $(a,v)$ associated with the queried address $a$ during lookup. Formally, for all $\lambda$, $X$, and each $(a,v)\in{X}$, we have

    \begin{equation}
        \Pr\left[\begin{array}{c}
        T\leftarrow{Build}(1^{\lambda},X),T\neq{\bot}:\\
        {Lookup(a,T)=(a,v)}
        \end{array}\right]=1
    \end{equation}

\end{definition}

We present the following theorem regarding the build failure probability of cuckoo hashing, as proposed in \cite{cryptoeprint:2021/447}.

\begin{theorem}[Build failure probability of cuckoo hashing \cite{cryptoeprint:2021/447}]
\normalfont
    Given an integer $N$, the build failure probability for a two-table cuckoo hashing scheme with $n = \omega(\log{N})$ elements and a stash of size $s = \Theta(\log{N})$ is negligible in $N$.
\end{theorem}

In the context of our ORAM constructions, $N$ represents the size of the entire database. This theorem indicates that for $\Omega(\log^{2}{N})$ elements, setting the cuckoo table size to $\Omega(\log^{2}{N})$ and the stash size to $\Theta(\log{N})$ ensures a negligible failure probability with respect to $\lambda = N$. This property is leveraged in our scheme, as the cuckoo table size for each level of our ORAM structure is $\Omega(\log^{2}{N})$.

\subsection{Distributed Point Function}
Distributed point function (DPF) was introduced by Gilboa et al. \cite{DBLP:conf/eurocrypt/GilboaI14} as a cryptographic primitive for secure searching in distributed settings. For a positive integer $c$, the point function $P_{i,x}:\{0,1\}^{|c|}\rightarrow{\{0,1\}^{|x|}}$ is defined such that $P_{i,x}(i) = x$ and $P_{i,x}(i') = 0^{|x|}$ for all $i' \neq i$. This paper primarily considers the single-bit point function \cite{DBLP:conf/ccs/BoyleGI16}, where the output domain of $P_{i,x}$ is $\{0,1\}$. A DPF provides function secret sharing for the point function.

\begin{definition}[Syntax of DPF]\label{SynDPF}
\normalfont
    A DPF protocol $\Pi_{DPF}$ consists of two algorithms $(Gen, Eval)$, with the following syntax.

    \begin{itemize}[leftmargin=*]
        \item $Gen(1^{\kappa},i,n,x)$: A PPT algorithm running on a client that is given a security parameter ${\kappa}$ and parameters $(i,n,x)$ describing a point function $P_{i,x}$ with the input domain $n$, outputs $k_{0}$ and $k_{1}$.

        \item $Eval(b,k_{b},i')$: A PPT algorithm running on two servers that on input the party indicator $b\in\{0,1\}$, a key $k_{b}$, and a point position $i'$, outputs a value belonging to the output domain. 
    \end{itemize}
        
    The \textbf{correctness} requires that for all $\kappa,i,n,x$, and each $i'\in\{0,\ldots,n-1\}$, we have

    \begin{equation}
        \Pr\left[\begin{array}{c}
        (k_{0},k_{1})\leftarrow{Gen}(1^{\kappa},i,n,x):\\
        \sum_{b=0}^{1}{Eval(b,k_{b},i')=P_{i,x}(i')}
        \end{array}\right] = 1
    \end{equation}
    
    The \textbf{security} ensures that in a non-colluding two-server setting, neither server can infer the description of $P_{i,x}$. In other words, given the security parameter $\kappa$, for all $n$, two values $x$ and $y$, two indexes $i\in\{0,\ldots,n-1\}$ and $j\in\{0,\ldots,n-1\}$, and a bit $b\in\{0,1\}$, the following distribution
    \begin{equation}
        \left\{(k_{0},k_{1})\leftarrow{Gen}(1^{\kappa},i,n,x): k_{b}\right\}
    \end{equation}
    is computationally indistinguishable with 
    \begin{equation}
        \left\{(k_{0},k_{1})\leftarrow{Gen}(1^{\kappa},j,n,y): k_{b}\right\}.
    \end{equation}

\end{definition}

Boyle et al. \cite{DBLP:conf/eurocrypt/BoyleGI15,DBLP:conf/ccs/BoyleGI16,DBLP:conf/eurocrypt/BoyleCG0I0R21} proposed several optimized DPF constructions that require only $O(B + \kappa \log{n})$ bits of communication and $O(\log{n})$ symmetric key computations for block size $B$. This implies that the bandwidth overhead is $O(1)$ for $B = O(\kappa \log{n})$.

\subsection{PIR in Two-server Settings}
We divide two-server PIR into read-only PIR and write-only PIR \cite{DBLP:conf/pkc/KushilevitzM19} according to different functionalities.

\noindent \textbf{Read-only PIR.} A read-only PIR ensures that a client can retrieve required values from an array without revealing to the servers which specific values were accessed.

\begin{definition}[Syntax of Read-only PIR]
\normalfont
    A two-server read-only PIR protocol $\Pi_{RPIR}$ contains two processes $(Setup, Read)$.

    \begin{itemize}[leftmargin=*]
        \item $Setup(1^\kappa,X)$: A setup algorithm that is given a security parameter $\kappa$ and an array $X$ with $n$ values of $B$ bits each from the client, outputs two arrays $X_{0},X_{1}$ stored on two servers, respectively.

        \item $Read(i,X_{0},X_{1})$: A read algorithm inputs an index $i$ from the client and arrays $X_{0},X_{1}$ from two servers. The client sends a pair of queries $q_{0},q_{1}$ to two servers and receives responses $r_{0},r_{1}$ from two servers. The client then gets the result by computing $r_{0}\oplus{r_{1}}$.
    \end{itemize}

    The \textbf{correctness} requires that for all $\kappa,X$ and $i\in\{0,\ldots,|X|-1\}$, we have
    \begin{equation}
        \Pr\left[\begin{array}{c}(X_{0},X_{1})\leftarrow{Setup(1^\kappa,X)};\\
        \big\{(q_{0},r_{0}),(q_{1},r_{1})\big\}\leftarrow{Read}(i,X_{0},X_{1}):r_{0}\oplus{r_{1}}=X[i]
        \end{array}\right] = 1.
    \end{equation}

    The \textbf{security} requires that given the security parameter $\kappa$, for all $X$, two indexes $i\in\{0,\ldots,|X|-1\},j\in\{0,\ldots,|X|-1\}$, and a bit $b\in\{0,1\}$, the following distribution
    \begin{equation}
        \left\{\begin{array}{c}(X_{0},X_{1})\leftarrow{Setup(1^\kappa,X)};\\
        \big\{(q_{0},r_{0}),(q_{1},r_{1})\big\}\leftarrow{Read}(i,X_{0},X_{1}):q_{b}
        \end{array}\right\}
    \end{equation}
    is computationally indistinguishable with 
    \begin{equation}
        \left\{\begin{array}{c}(X_{0},X_{1})\leftarrow{Setup(1^\kappa,X)};\\
        \big\{(q_{0},r_{0}),(q_{1},r_{1})\big\}\leftarrow{Read}(j,X_{0},X_{1}):q_{b}
        \end{array}\right\}.
    \end{equation}
\end{definition}

\begin{figure*}[t!]
    \centering
    \subfigure[Setup phase]
    {
        \begin{minipage}{.42\linewidth}
            \centering
            \caption{Setup phase}
            \includegraphics[width=\textwidth]{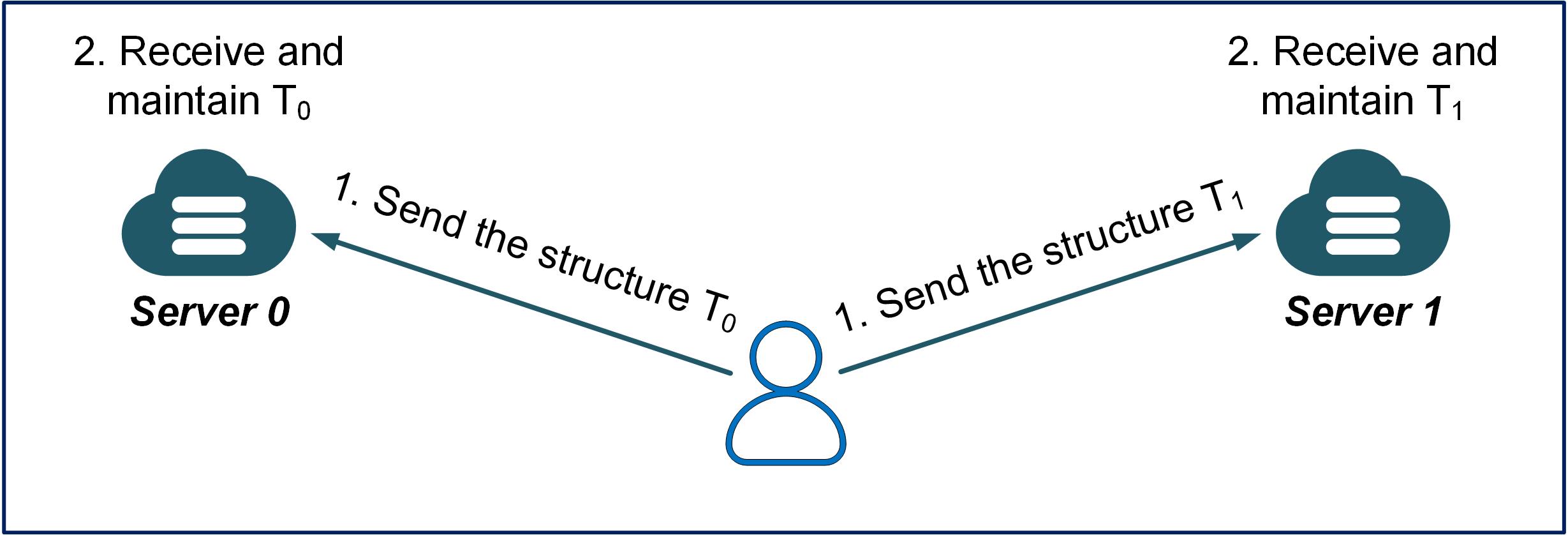}
            \label{fig:Setup}
        \end{minipage}    
    }
    \quad\quad
    \subfigure[Accesss phase]
    {
        \begin{minipage}{.42\linewidth} 
            \centering
            \includegraphics[width=\textwidth]{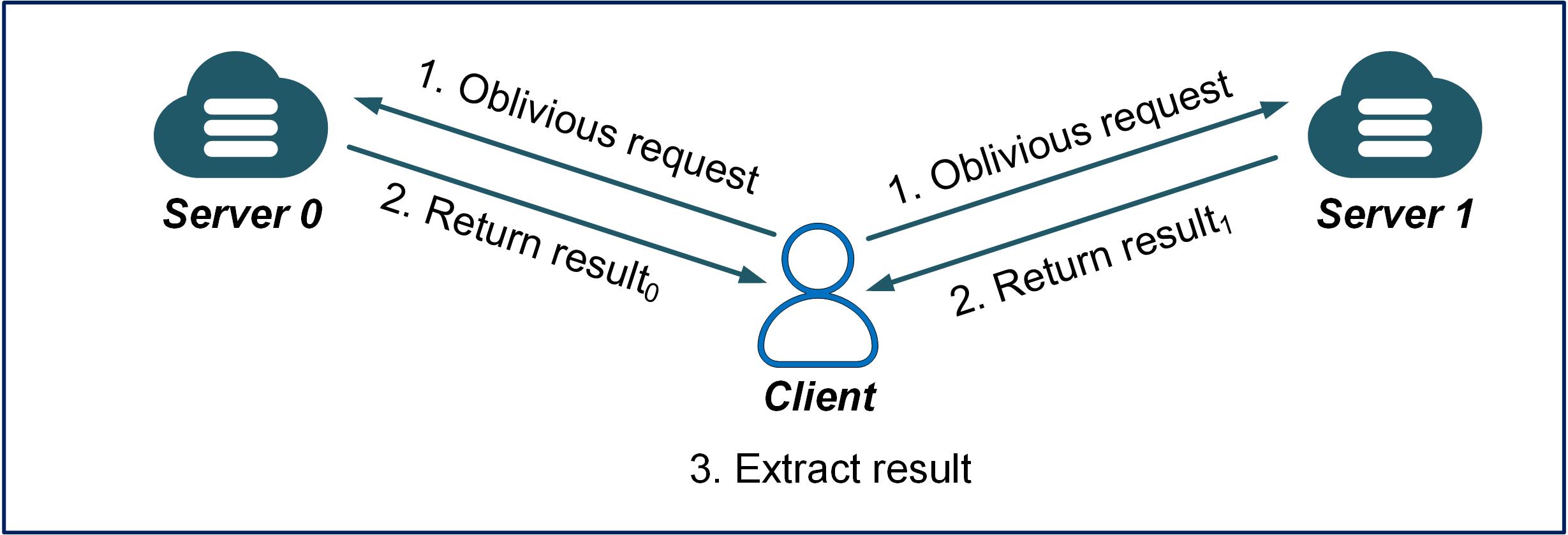}
            \label{fig:Access}
        \end{minipage}
    }
    \caption{
    Two-server ORAM system model. In the setup phase, the client sends the distributed structures to two servers. In the access phase, the client requests an element and retrieves results from two servers. The client may reconstruct partial structures during the access phase, e.g., evict a path for tree-based ORAM or rebuild a level for hierarchical ORAM.
    }
    \label{fig:SystemModel}
\end{figure*}

\noindent \textbf{Write-only PIR.} A write-only PIR ensures that a client can privately write a value to a specific position of an array without revealing either the value being written or the position to the servers.

\begin{definition}[Syntax of Write-only PIR]
\normalfont
A write-only PIR protocol $\Pi_{WPIR}$ contains three processes $(Setup, Write, Build)$.

    \begin{itemize}[leftmargin=*]
        \item $Setup(1^\kappa,X)$: A setup algorithm that is given a security parameter $\kappa$, an array $X$ of $n$ values with $B$ bits each from the client, outputs two arrays $X_{0},X_{1}$ stored on two servers, respectively.

        \item $Write(i,x_{old},x_{new},X_{0},X_{1})$: A write algorithm inputs an index $i$, an old value $x_{old}$ currently stored in the $i$th position of $X$, a new value $x_{new}$ from the client, and arrays $X_{0},X_{1}$ from two servers. The client sends a query $q_{b}$ and a value $v_{b}$ to the server $\mathcal{S}_{b}$ for $b\in\{0,1\}$. The two servers then update arrays $X_{0},X_{1}$.

        \item $Build(X_{0},X_{1})$: A build algorithm that inputs two updated arrays $X_{0},X_{1}$ on two servers, outputs a new array.
    \end{itemize}

    The \textbf{correctness} requires that the client can correctly reconstruct the updated array $X$ after executing the $Build(X_{0},X_{1})$.
    Formally, for all $\kappa,X$, an index sequence $iS$, and a value sequence $vS$, where $|iS|=|vS|$, we have
    \begin{equation}
        \Pr\left[\begin{array}{c}(X_{0},X_{1})\leftarrow{Setup(1^\kappa,X)};\\
        \Big\{\big\{(q_{0},v_{0}),(q_{1},v_{1})\big\}
        \leftarrow\\
        {Write}(iS[k],x_{old},vS[k],X_{0},X_{1})\Big\}_{k\in\{0,\ldots,|iS|-1\}}:\\
        Build(X_{0},X_{1})=X'
        \end{array}\right] = 1,
    \end{equation}
    where $X'$ is an element array after executing the write operation sequence on $X$, i.e., $X'=\big\{X:\big\{X\big[iS[k]\big]=vS[k]\big\}_{k\in\{0,...,|iS|-1\}}\big\}$.

    The \textbf{security} requires that given the security parameter $\kappa$, for all $X$, two indexes $i\in\{0,..,|X|-1\}$ and $j\in\{0,..,|X|-1\}$, two values $y$ and $z$, and a bit $b\in\{0,1\}$, the following distribution
    \begin{equation}
        \left\{\begin{array}{c}(X_{0},X_{1})\leftarrow{Setup(1^\kappa,X)};\\
        \big\{(q_{0},v_{0}),(q_{1},v_{1})\big\}\leftarrow{Write}(i,x_{old},y,X_{0},X_{1}):(q_{b},v_{b})
        \end{array}\right\}
    \end{equation}
    is computationally indistinguishable with 
    \begin{equation}
        \left\{\begin{array}{c}(X_{0},X_{1})\leftarrow{Setup(1^\kappa,X)};\\
        \big\{(q_{0},v_{0}),(q_{1},v_{1})\big\}\leftarrow{Write}(j,x_{old},z,X_{0},X_{1}):(q_{b},v_{b})
        \end{array}\right\}.
    \end{equation}
    
\end{definition}

Note that Kushilevitz and Mour \cite{DBLP:conf/pkc/KushilevitzM19} described secure read-only PIR and write-only PIR instances in two-server settings. We present details in Appendix \ref{AppPIR}.

\section{System and Security Models}\label{SG}
In this section, we present our system model, illustrated in Fig. \ref{fig:SystemModel}, which involves two protocols between the client and two servers. We formally define the system as follows:

\begin{definition}[Two-server ORAM]
\normalfont

    A two-server ORAM protocol $\Pi_{ORAM}$ comprises two algorithms $(Setup, Access)$.

    \begin{itemize}[leftmargin=*]
        \item $Setup(1^{\lambda}, X)$: A setup algorithm takes as input a security parameter $\lambda$ and an array $X$ containing $N$ elements, each with $B$ bits. It outputs an encrypted storage structure $EX_b$, which is stored on server $S_b$, for $b \in \{0, 1\}$.

        \item $Access(op,a,writeV)$: An access algorithm takes as input a triple $(op, a, writeV)$, where the operation $op\in\{read,write\}$, the virtual address $a\in\{0,...,N-1\}$, and the value $writeV\in\{0,1\}^{B}$. It outputs $X[a]$ to the client if $op\equiv{read}$; otherwise, it returns the value $X[a]$ and updates the data on virtual address $a$ to the new value $writeV$. Moreover, it may update the storage structure $EX_{b}$ to $EX'_{b}$ for two servers.
    \end{itemize}
    
\end{definition}

The \textbf{correctness} requires that for any legal $(op,a,writeV)$, the access protocol always returns the latest data consistent with $X[a]$.

For \textbf{security}, we assume that the two non-colluding servers are honest but curious. They cannot maliciously tamper with data or deviate from the protocol but try to infer some sensitive information based on their observations. Our security goal is to ensure obliviousness and confidentiality. It ensures that the adversary cannot infer the requested address $a$ from the access patterns and cannot know the contents of elements.

We respectively describe executions in a real world and an ideal world in $\mathbf{Real}_{\mathcal{A},\mathcal{C}}^{\Pi_{ORAM}}$ and $\mathbf{Ideal}_{\mathcal{A},\mathcal{C}}^{Sim}$. In the real world, the client interacts with the adversary through a real protocol $\Pi_{ORAM}$. In the ideal world, the client invokes a simulator $Sim$, who only knows the input array size, to interact with the adversary. In both worlds, the adversary can observe the output elements $Out_{b}$, the updated storage structure $EX'_{b}$, and the memory access patterns $AP_{b}$ during the access phase. The adversary’s goal is to determine whether the execution is in the real world or in the ideal world.

{
\renewcommand{\thealgorithm}{}
\floatname{algorithm}{}

\begin{algorithm}[!t]\label{Alg. RealGame}
\caption{$\mathbf{Real}_{\mathcal{A},\mathcal{C}}^{\Pi_{ORAM}}$}
\label{RealLab}
\begin{algorithmic}[1]
    \State The adversary $\mathcal{A}$ chooses a bit $b\in\{0,1\}$ and corrupts the server $\mathcal{S}_{b}$.
    \State $\mathcal{A}$ randomly generates an element array by $X_{b}\gets{\mathcal{A}(1^\lambda)}$ and sends it to the client.
    \State The client $\mathcal{C}$ receives $X_{b}$, constructs $EX_{b}\gets\Pi_{ORAM}.Setup(1^{\lambda},X_{b})$, and sends the $EX_{b}$ to $\mathcal{A}$.
    \State The adversary $\mathcal{A}$ generates $(op,a,writeV)\gets{\mathcal{A}(1^\lambda,EX_{b})}$ and sends it to the client.
    \State\textit{Loop while} {$op\neq{\perp}$:}
        \State \quad \quad $(Out_{b},EX'_{b},AP_{b})\gets\Pi_{ORAM}.Access(op,a,writeV)$.
        \State \quad \quad The adversary $\mathcal{A}$  generates $(op,a,writeV)\gets{\mathcal{A}(1^\lambda,Out_{b},EX'_{b},AP_{b})}$ and sends it to the client.
    \State The adversary $\mathcal{A}$ outputs a bit $q\in\{0,1\}$.
\end{algorithmic}
\end{algorithm}

\begin{algorithm}[!t]\label{Alg. IdealGame}
\caption{$\mathbf{Ideal}_{\mathcal{A},\mathcal{C}}^{Sim}$\label{IdealLab}}
\begin{algorithmic}[1]
    \State The adversary $\mathcal{A}$ chooses a bit $b\in\{0,1\}$ and corrupts the server $\mathcal{S}_{b}$.
    \State $\mathcal{A}$ randomly generates an element array by $X_{b}\gets{\mathcal{A}(1^\lambda)}$ and sends it to the client.
    \State The client $\mathcal{C}$ receives $X_{b}$, constructs $EX_{b}$ by executing simulator $EX_{b}\gets{Sim}(1^{\lambda},|X_{b}|)$, and sends the $EX_{b}$ to $\mathcal{A}$.
    \State The adversary $\mathcal{A}$ generates $(op,a,writeV)\gets{\mathcal{A}(1^\lambda,EX_{b})}$ and sends it to the client.
    \State\textit{Loop while} {$op\neq{\perp}$:}
        \State \quad \quad $(Out_{b},EX'_{b},AP_{b})\gets{Sim}(1^{\lambda},|X_{b}|)$.
        \State \quad \quad The adversary $\mathcal{A}$ generates $(op,a,writeV)\gets{\mathcal{A}(1^\lambda,Out_{b},EX'_{b},AP_{b})}$ and sends it to the client.
    \State The adversary $\mathcal{A}$ outputs a bit $q\in\{0,1\}$.
\end{algorithmic}
\end{algorithm}
}

\begin{definition}[Oblivious simulation security]
We say that a two-server ORAM protocol $\Pi_{ORAM}$ satisfies the oblivious simulation security in the presence of an adversary $\mathcal{A}$ corrupting only one server $\mathcal{S}_{b}$ where $b\in\{0,1\}$ iff for any PPT real-world adversary $\mathcal{A}$, there exists a simulator $Sim$, such that 
\begin{equation}
    \big|Pr[\mathbf{Real}_{\mathcal{A},\mathcal{C}}^{\Pi_{ORAM}}=1]-Pr[\mathbf{Ideal}_{\mathcal{A},\mathcal{C}}^{Sim}=1]\big|\leq{negl(\lambda)}
\end{equation}
\end{definition}

\section{The Proposed Scheme Cforam}\label{SecScheme}

In this section, we introduce the storage structure, detailed construction\footnote{Note that in our protocol, we use \textit{encrypt} and \textit{decrypt} to represent symmetric encryption operations, such as AES, for encrypting and decrypting elements.}, and overhead of Cforam. The security of Cforam is ensured by the pseudorandomness of PRF, the security of PIR, and the negligible build failure probability of cuckoo hashing. A thorough analysis is provided in Appendix \ref{AppCforam}.

\subsection{ORAM Storage Structures}

\noindent \textbf{Server $\mathcal{S}_{b}$'s storage.} For $N$ elements, each server $\mathcal{S}_{b}$'s storage $EX_{b}$ contains a pairwise area separately constructed by replicated elements and secret-shared tags in $L-\ell+1$ levels, numbered $\{\ell,\ell+1,\ldots,L\}$, where $\ell=\lceil{2\log\log{N}}\rceil$ and $L=\lceil{\log{N}}\rceil$. The server's storage structure is illustrated in Fig. \ref{OurSto}. 

The element area includes:

\begin{figure}[!t]
    \centering
    \includegraphics[width=\linewidth]{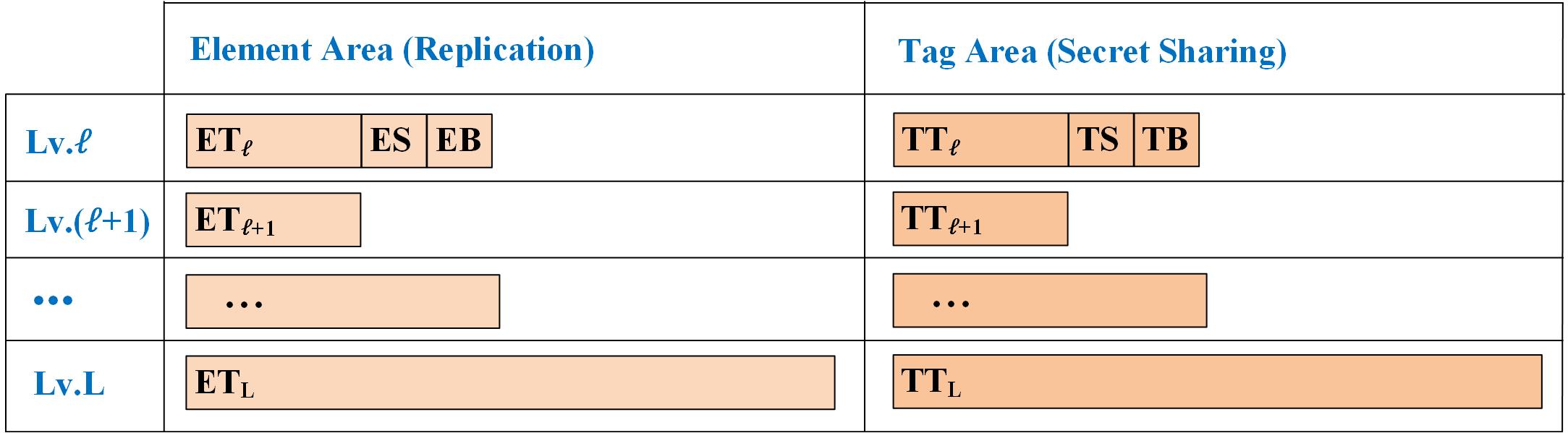}
    \caption{The server storage of our construction. Each server stores a replicated element area and a secret-shared tag area spanning $L-\ell+1$ levels. For each level $i\in\{\ell,\ldots,L\}$, the server storage includes element tables $ET_{i}=(ET_{i0},ET_{i1})$ and tag tables $TT_{i}=(TT_{i0},TT_{i1})$. Additionally, the $\ell$th level contains an element stash and a buffer, a tag stash and a buffer.}
    \label{OurSto}
\end{figure}

\begin{itemize}[leftmargin=*]
    \item Lv. $\ell$: Two tables $\langle{ET}_{\ell0}\rangle_{b}$, $\langle{ET}_{\ell1}\rangle_{b}$, and a stash $\langle{ES}\rangle_{b}$ using cuckoo hashing, with the capacity $n=2^{\ell+1}+\log{N}\cdot{(L-\ell)}$; A buffer $\langle{EB}\rangle_{b}$ with $\log{N}$ capacity. Each position of these structures stores an element $(a, v)$.
    
    \item Lv. $i\in\{\ell+1,...,L\}$: Two tables $\langle{ET}_{i0}\rangle_{b}$ and $\langle{ET}_{i1}\rangle_{b}$ for cuckoo hashing, which can accommodate $n=2^{i}$ elements.
\end{itemize}

The tag area includes:
\begin{itemize}[leftmargin=*]
    \item Lv. $\ell$: Two tables $\llbracket{TT}_{\ell0}\rrbracket_{b}$, $\llbracket{TT}_{\ell1}\rrbracket_{b}$, and a stash $\llbracket{TS}\rrbracket_{b}$ for cuckoo hashing, with the capacity $n=2^{\ell+1}+\log{N}\cdot{(L-\ell)}$; A buffer $\llbracket{TB}\rrbracket_{b}$ with $\log{N}$ capacity. Each position records a secret-shared tag $\tau$ given a virtual address $a$.

    \item Lv. $i\in\{\ell+1,...,L\}$: Two tables $\llbracket{TT}_{i0}\rrbracket_{b}$, and $\llbracket{TT}_{i1}\rrbracket_{b}$ for cuckoo hashing with the capacity $n=2^{i}$, where each position stores a secret-shared tag $\tau$.
\end{itemize}

\begin{remark}
\normalfont
    In our structure, the first positions (i.e., position $0$) of tables, stashes, and buffers do not store real elements and are applied in \textbf{PIR pseudo-write} operations. A PIR pseudo-write indicates that, given an array-like storage structure, the client invokes write-only PIR to obliviously modify the content at position $0$.
\end{remark}

\begin{remark}
\normalfont
    We classify the stored elements into three types: \textbf{real elements, dummy elements, and empty elements}. During the setup phase, the elements written into the storage structure are real elements. During the access phase, if the client locates the required element at a particular position, (s)he will mark the element at that location as a dummy by overwriting its tag. Empty elements in the hash table indicate that the corresponding position is unoccupied. When rebuilding the $i$th level ($i \neq L$), the client keeps real and dummy elements, while removing only empty elements; when rebuilding the $L$th level, the client preserves real elements in this level, while removing both dummy and empty elements.
\end{remark}

\noindent \textbf{Client storage.} The client preserves:

\begin{itemize}[leftmargin=*]
    \item One bit $full_{i}$ for level $i\in\{\ell,...,L\}$, where $1$ indicates the level contains non-empty elements and $0$ denotes that the level is empty.

    \item Two indicators $lenS$ and $lenB$, which denote the number of elements in the stash and the buffer, respectively.

    \item A global counter $ctr$ initialized to $0$. After each access, the counter is incremented by $1$. Using this counter, the client can locally calculate the $epoch_{i}$ for each level $i\in\{\ell,...,L\}$, which represents the number of times the $i$th level has been rebuilt.

    \item A level master key $lk$ for generating hash function keys of each level and a tag key $tk$ for generating tags.
\end{itemize}

\subsection{Full Construction}
We now formally describe Cforam, which consists of three functions: Setup, Access, and Rebuild. 

The overall architecture is depicted in Fig. \ref{OurArx}. Cforam includes both the setup protocol $\Pi_{Setup}$ and the access protocol $\Pi_{Access}$. After a fixed number of accesses, $\Pi_{Access}$ triggers the rebuild protocol $\Pi_{Rebuild}$ to reconstruct a level. The setup and rebuild protocols require the cuckoo hashing $\Pi_{CuckHash}$ to construct tables. During access, the protocol $\Pi_{Access}$ will leverage the DPF-based read-only PIR $\Pi_{RPIR}$ and write-only PIR $\Pi_{WPIR}$ to handle client requests.

\begin{figure}[!t]
    \centering
    \includegraphics[height=0.6\linewidth]{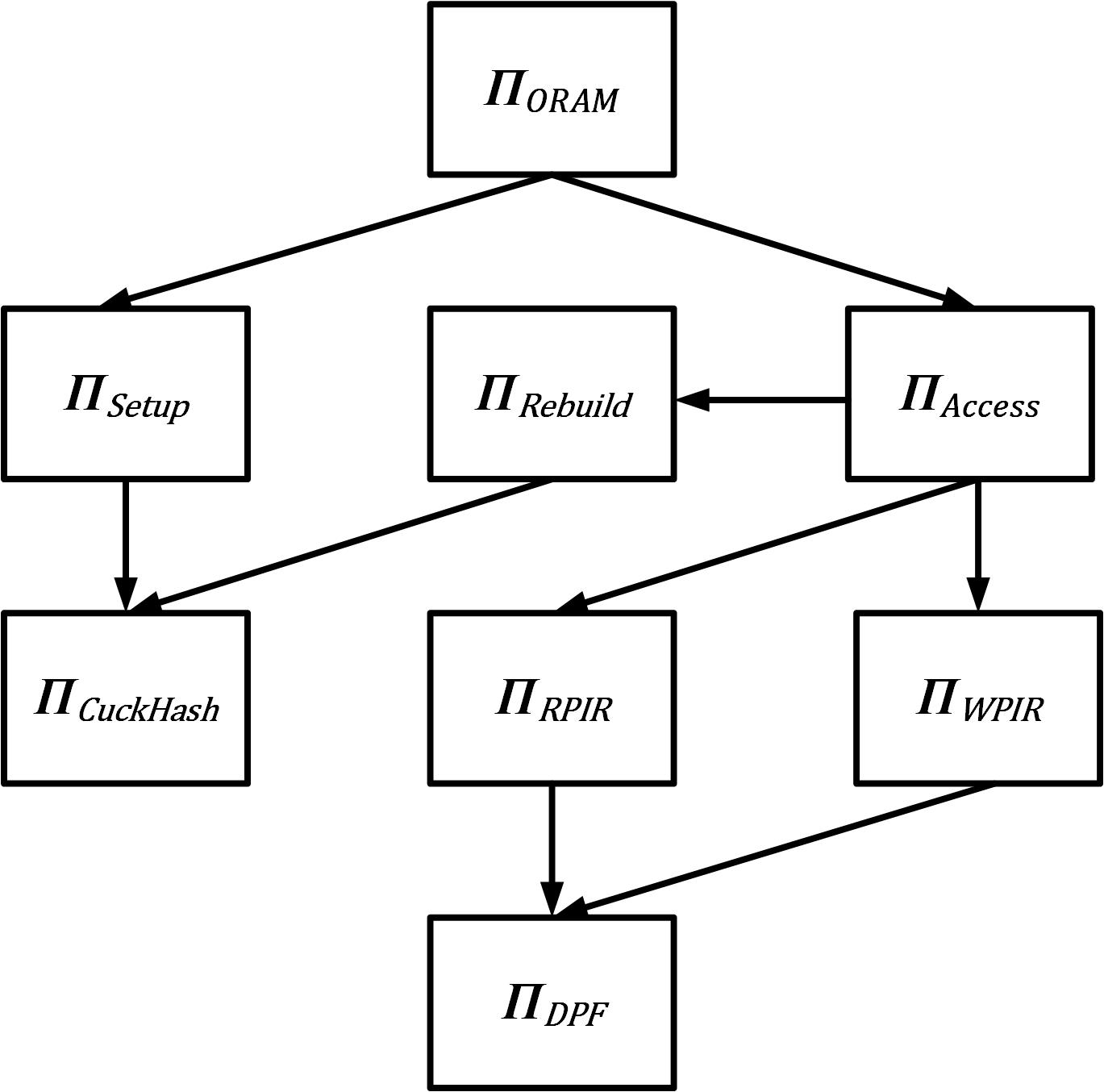}
    \caption{The overall architecture of our ORAM constructions.}
    \label{OurArx}
\end{figure} 

\smallskip

\begin{algorithm}[!t]
    \renewcommand{\thealgorithm}{1}
    \floatname{algorithm}{Protocol}
    \caption{$\Pi_{Setup}(1^\lambda,X)$}\label{AlgoInitial}
    \begin{algorithmic}[1]
        
        \Statex \underline{\textit{Client:}}
        \smallskip

        \State Initialize $full_{i}\coloneqq{0}$ for each $i\in\{\ell,...,L\}$, $lenS\coloneqq{1}$, $lenB\coloneqq{1}$, $ctr\coloneqq{0}$, a secret level master key $lk$, and a tag key $tk$.
        \State Compute the $\ell$th level hash function key $hk_{\ell}\leftarrow{F}(lk,\ell,epoch_{\ell})$ and $L$th level hash function key $hk_{L}\leftarrow{F}(lk,L,epoch_{L})$.

        \For{each $(a,v)\in{X}$}
            Invoke \textbf{Insert}($a,v,lev=L$) to place the encrypted element to the ORAM structures of the servers.
        \EndFor
        \smallskip
        
        \Statex \underline{\textit{Server $\mathcal{S}_{0}$:}}
        \smallskip
        \State {\color{red}Send to client} the message $m=(0,1)$ if the $\ell$th table has no real elements; otherwise,  {\color{red}send to client} $m=(1,|\langle{ES}\rangle_0|)$.
        \smallskip
        
        \Statex \underline{\textit{Client:}}
        \smallskip
        \State \textit{Receive} from server $\mathcal{S}_{0}$ the message $m$.
        \State Set $full_{L}\coloneqq{1}$, $full_{\ell}\coloneqq{m}[0]$, $lenS\coloneqq{m}[1]$.
    \end{algorithmic}
\end{algorithm}

\begin{algorithm}[!t]
    \renewcommand{\thealgorithm}{}
    \floatname{algorithm}{Algorithm}
    \caption{\textbf{Insert}($a,v,lev$)}
    \begin{algorithmic}[1]
        \Statex \underline{\textit{Client:}}
        \smallskip
        \State Compute the tag $\tau\leftarrow{F}(tk,a)$.
        \State Compute positions $(pos_{lev0},pos_{lev1})\leftarrow{H}(hk_{lev},\tau)$ for $lev$th table and $(pos_{\ell{0}},pos_{\ell{1}})\leftarrow{H}(hk_{\ell},\tau)$ for $\ell$th table.
        \State Generate $\tau_{0}$ and $\tau_{1}$ satisfying that $\tau_{0}\oplus\tau_{1}=\tau$.
        \State \textit{Encrypt} the element $(a,v)$, {\color{red}send to servers $\mathcal{S}_{b}$ for $b\in\{0,1\}$} the encrypted element $(a,v)$, its secret-shared tag $\tau_{b}$, and positions $(pos_{lev0},pos_{lev1},pos_{\ell{0}},pos_{\ell{1}})$.

        \Statex \underline{\textit{Servers $\mathcal{S}_{b}$:}}
        \smallskip
        \State \textit{Receive} and try to place the encrypted $(a,v)$ into the cuckoo hashing table of element area and $\tau_{b}$ into the table of tag area at the $lev$th level.
        \If{$lev$th level overflows}
            \State Place the encrypted $(a,v)$ into the table of element area and $\tau_{b}$ into the table of tag area at the $\ell$th level.
        \EndIf
    \end{algorithmic}
\end{algorithm}

\begin{figure*}[!t]
    \centering
    \includegraphics[width=0.85\linewidth]{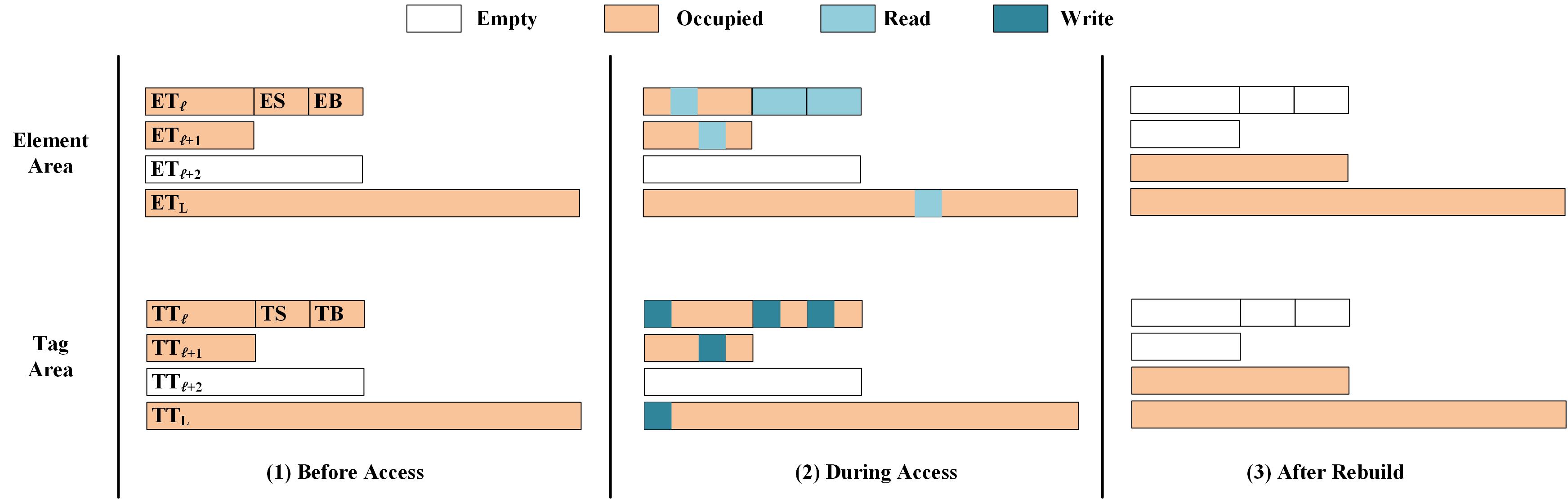}
    \caption{\normalsize{An example of Cforam with four levels for different phases. The diagram uses following colors to represent the components: (i) \textbf{White}: Empty levels. (ii) \textbf{Orange}: Non-empty levels that are occupied. (iii) \textbf{Light Blue}: Positions for element reading. (iv) \textbf{Dark Blue}: Positions for tag writing.}}
    \label{OurExample}
\end{figure*}

\begin{algorithm*}[!t]
    \renewcommand{\thealgorithm}{2}
    \floatname{algorithm}{Protocol}
    \caption{$\Pi_{Access}(op,a,writeV)$}\label{AlgoAccess}

    \begin{multicols}{2}
    \begin{algorithmic}[1]
        \Statex \underline{\textit{Server $\mathcal{S}_{0}$:}}
        \smallskip
        \For{$ind=|\langle{EB}\rangle_{0}|-1,...,1$}
            {\color{red}Send to client} 
            $\langle{EB}\rangle_{0}[ind]$. \label{AlAccEBBegin}\label{AlgoAcLine1}
        \EndFor
        \For{$ind=|\langle{ES}\rangle_{0}|-1,...,1$}\label{AlAccESBegin}
            {\color{red}Send to client} $\langle{ES}\rangle_{0}[ind]$.\label{AlgoAcLine2}
        \EndFor
        
        \smallskip
        
        \Statex \underline{\textit{Client:}}
        \smallskip
        \State Set $found\coloneqq{false}$, $posEB\coloneqq{0}$, $posES\coloneqq{0}$, $recE\coloneqq{(\bot,\bot)}$.
        \State Compute a tag $\tau\leftarrow{F}(tk,a)$ and random non-zero tag $\tau_{rand}$.
        \For{$ind=lenB-1,...,1$}
            \State \textit{Receive} and \textit{decrypt} each element $(rA,rV)$ from the buffer $\langle{EB}\rangle_{0}$ of server $\mathcal{S}_{0}$.
            \If{$rA\equiv{a}$ $\wedge$ ${\neg{found}}$}
                $recE\coloneqq{(rA,rV)}$, $found\coloneqq{true}$, $posEB\coloneqq{ind}$.
            \EndIf
        \EndFor 
        \State \textbf{end for}
        \For{$ind=lenS-1,...,1$}
            \State \textit{Receive} and \textit{decrypt} each element $(rA,rV)$ from the stash $\langle{ES}\rangle_{0}$ of server $\mathcal{S}_{0}$.
            \If{$rA\equiv{a}$ $\wedge$ ${\neg{found}}$}
                $recE\coloneqq{(rA,rV)}$, $found\coloneqq{true}$, $posES\coloneqq{ind}$.
            \EndIf
        \EndFor
        \State \textbf{end for}\label{AlgoAcLine12}
        \State Invoke $\mathbf{\Pi_{WPIR}.Write}(posEB,\tau,\tau\oplus\tau_{rand},\llbracket{TB}\rrbracket_{0},\llbracket{TB}\rrbracket_{1})$ to privately modify the tag buffer on two servers.\label{AlgoAcLine13}
        \State Invoke $\mathbf{\Pi_{WPIR}.Write}(posES,\tau,\tau\oplus\tau_{rand},\llbracket{TS}\rrbracket_{0},\llbracket{TS}\rrbracket_{1})$ to privately modify the tag stash on two servers. \label{AlgoAcLine14}
        
        \For{level $i=\ell,...,L$}\label{AlAccETBegin}\label{AlgoAcLine15}
            \If{$full_{i}\equiv{0}$} continue.
            \EndIf
            \State Compute the $i$th level hash key $hk_{i}\gets{F(lk,i,epoch_{i})}$.
            \State Set $(posR_{i0},posR_{i1})\gets{H(hk_{i},\tau)}$ for element reading and $(posW_{i0}\coloneqq{0}$, $posW_{i1}\coloneqq{0})$ for tag writing.
            \State Invoke $\mathbf{\Pi_{RPIR}.Read}(posR_{i0},\langle{ET}_{i0}\rangle_{0},\langle{ET_{i0}}\rangle_{1})$ to get $r_{i0}$ and $\mathbf{\Pi_{RPIR}.Read}(posR_{i1},\langle{ET}_{i1}\rangle_{0},\langle{ET_{i1}}\rangle_{1})$ to get $r_{i1}$.\label{AlgoAcLine19}  
            \State \textit{Decrypt} the elements $r_{i0}$ and $r_{i1}$.
            
            \If{$r_{i0}[0]\equiv{a}$ $\wedge$ ${\neg{found}}$}
                \State $recE\coloneqq{r_{i0}}$, $found\coloneqq{true}$, $posW_{i0}\coloneqq{posR_{i0}}$.. 
            \EndIf
            \If{$r_{i1}[0]\equiv{a}$ $\wedge$ ${\neg{found}}$}
                \State $recE\coloneqq{r_{i1}}$, $found\coloneqq{true}$, $posW_{i1}\coloneqq{posR_{i1}}$. 
            \EndIf\label{AlgoAcLine23}

        
            \State Invoke  $\mathbf{\Pi_{WPIR}.Write}(posW_{i0},\tau,\tau\oplus\tau_{rand},\llbracket{TT}_{i0}\rrbracket_{0},\llbracket{TT}_{i0}\rrbracket_{1})$ and $\mathbf{\Pi_{WPIR}.Write}(posW_{i1},\tau,\tau\oplus\tau_{rand},\llbracket{TT}_{i1}\rrbracket_{0},\llbracket{TT}_{i1}\rrbracket_{1})$ to privately modify tags in two hash tables of the $i$th level.\label{AlgoAcLine24}
        \EndFor
        \State \textbf{end for}

        \State Generate $\tau_{0},\tau_{1}$ that $\tau_{0}\oplus\tau_{1}=\tau$ and let $lenB\coloneqq{lenB+1}$.
        \State {Set $sendE\coloneqq{(a,writeV)}$ if $op\equiv{write}$ or $sendE\coloneqq(recE[0],recE[1])$ otherwise.}
        \State \textit{Encrypt} the element $SendE$, {\color{red}send to servers $\mathcal{S}_{b}$} the encrypted element $sendE$ and tag $\tau_{b}$.\label{AlgoAcLine28}
        \smallskip

        \Statex \underline{\textit{Servers $\mathcal{S}_{b}$:}}
        \smallskip
        \State {\textit{Receive} the encrypted element $SendE$ and tag $\tau_{b}$, then add $SendE$ to $\langle{EB}\rangle_{b}$ and $\tau_{b}$ to $\llbracket{TB}\rrbracket_{b}$.}\label{AlAccETEnd}
        \smallskip
        
        \Statex \underline{\textit{Client:}}
        \smallskip
        \State {Increment the global access counter $ctr$.}\label{AlAccRebBegin}
        \If{{$ctr\equiv{0\ mod\ {2^{L}}}$}} {Invoke $\mathbf{\Pi_{RebuildL}}$.}
        \EndIf
        \State {\textbf{elif} $ctr\equiv{0\ mod\ {2^{\ell+1}}}$, \textbf{then} find the smallest level $j\in\{\ell+1,...,L-1\}$ where $full_{j}\equiv{0}$ and invoke $\mathbf{\Pi_{Rebuild}}(j)$.}
        \State {\textbf{elif} $ctr\equiv{0\ mod\ {\log{N}}}$, \textbf{then} invoke $\mathbf{\Pi_{Rebuild}}(\ell)$.}\label{AlAccRebEnd}
    \end{algorithmic}
    
    \end{multicols}
\end{algorithm*}

\begin{algorithm}[!t]
    \renewcommand{\thealgorithm}{3}
    \floatname{algorithm}{Protocol}
    
    \caption{$\Pi_{Rebuild}(j)${\color{blue}\Comment{If $j\equiv{\ell}$, execute the algorithm except the dashed line; otherwise, execute the whole process.}}}\label{AlgoRELL}

    \begin{algorithmic}[1]
        \Statex \underline{\textit{Server $\mathcal{S}_{0}$:}}
        \smallskip
        \For{each non-empty encrypted element $ele=(a,v)$ and tag $\tau_{0}$ from the buffer, stash, and $\ell$th hash table}\label{AlRebL1}\label{AlRebJ1}
            \State {\color{red}Send to client} $(ele,\tau_{0})$.
        \EndFor
        
        \For{\dashuline{each non-empty encrypted element $ele=(a,v)$ and tag $\tau_{0}$ from hash tables of levels $\{\ell+1,...,j-1\}$}}
            \State \dashuline{{\color{red}Send to client} $(ele,\tau_{0})$.}
        \EndFor

        \smallskip
        
        \Statex \underline{\textit{Server $\mathcal{S}_{1}$:}}
        \smallskip
        \For{each non-empty $\tau_{1}$ from the buffer, stash, and $\ell$th hash table}
            \State {\color{red}Send to client} $\tau_{1}$.\label{AlRebL6}
        \EndFor
        
        \For{\dashuline{each non-empty $\tau_{1}$ from hash tables of levels $\{\ell+1,...,j-1\}$}}
            \State \dashuline{{\color{red}Send to client} $\tau_{1}$.}\label{AlRebL8}\label{AlRebJ8}
        \EndFor
        
        \smallskip
        
        \Statex \underline{\textit{Client:}}
        \smallskip
        \State \textit{Receive} and \textit{decrypt} each $ele=(a,v)$ and $\tau_{0}$ from the server $\mathcal{S}_{0}$ and $\tau_{1}$ from the server $\mathcal{S}_{1}$.\label{AlRebL9}
        \State \textbf{If} $ele\neq{dummyE}$ $\wedge$ $\tau_{0}\oplus\tau_{1}\neq{F(tk,a)}$, set $ele\coloneqq{dummyE}$.{\color{blue}\Comment{Note that $dummyE$ denotes the dummy element.}}\label{AlRebL10}
        \For{each $ele=(a,v)$}
            Invoke \textbf{Insert}($a,v,j$) to place the element at the $j$th level of the ORAM structures.\label{AlRebL13}
        \EndFor

        \State Set $full_{i}\coloneqq{0}$ for $i\in\{\ell,...,j-1\}$, $full_{j}\coloneqq{1}$,  $lenS\coloneqq{1}$, $lenB\coloneqq{1}$, and change $full_{\ell}$ and $lenS$ if the $\ell$th level is non-empty.\label{AlRebL14}

    \end{algorithmic}
\end{algorithm}

\begin{algorithm}[!t]
    \renewcommand{\thealgorithm}{4}
    \floatname{algorithm}{Protocol}
    \caption{$\Pi_{RebuildL}$}\label{AlgoRL}
    \begin{algorithmic}[1]
        
        \Statex \underline{\textit{Server $\mathcal{S}_{0}$:}}
        \smallskip
        \State {\color{red}Send to client} each non-empty encrypted element $ele=(a,v)$ and tag $\tau_{0}$.\label{AlgoRLLine1}
        \smallskip
        
        \Statex \underline{\textit{Server $\mathcal{S}_{1}$:}}
        \smallskip
        \State {\color{red}Send to client} each non-empty tag $\tau_{1}$.\label{AlgoRLLine2}
        
        \Statex \underline{\textit{Client:}}
        \smallskip
        \State \textit{Receive} and \textit{decrypt} each element $ele=(a,v)$, tag $\tau_{0}$ and $\tau_{1}$.\label{AlgoRLLine3}
        \State \textbf{If} $ele\neq{dummyE}$ $\wedge$ $\tau_{0}\oplus\tau_{1}\neq{F(tk,a)}$, \textbf{then} \textit{encrypt} a $dummyE$ and {\color{red}send it to server $\mathcal{S}_{0}$}; \textbf{otherwise}, \textit{encrypt} the $ele$ and {\color{red}send it to $\mathcal{S}_{0}$}.\label{AlgoRLLine4}
        \smallskip

        \Statex \underline{\textit{Server $\mathcal{S}_{0}$:}}
        \smallskip
        \State \textit{Receive} all encrypted elements from the client and \textit{shuffle} them.\label{AlgoRLLine5}
        \State {\color{red}Send to client} each encrypted element.\label{AlgoRLLine6}
        \smallskip

        \Statex \underline{\textit{Client:}}
        \smallskip
        \State \textit{Receive} and \textit{decrypt} each element $ele=(a,v)$ from $\mathcal{S}_{0}$.\label{AlgoRLLine7}
        \State \textbf{If} $ele\neq{dummyE}$, \textbf{then} \textit{encrypt} the $ele$ and {\color{red}send it to $\mathcal{S}_{1}$}.\label{AlgoRLLine8}
        \smallskip

        \Statex \underline{\textit{Server $\mathcal{S}_{1}$:}}
        \smallskip
        \State \textit{Receive} all encrypted elements from the client and \textit{shuffle} them.\label{AlgoRLLine9}
        \State {\color{red}Send to client} each encrypted element.\label{AlgoRLLine10}
        \smallskip

        \Statex \underline{\textit{Client:}}
        \smallskip
        \State Generate a new level master key and new tag key.\label{AlgoRLLine11}
        \State \textit{Receive} and \textit{decrypt} each element $ele=(a,v)$ from $\mathcal{S}_{1}$ and then invoke \textbf{Insert}$(a,v,L)$ to reconstruct the ORAM.\label{AlgoRLLine12}
        \State Set $full_{i}\coloneqq{0}$ for $i\in\{\ell,...,L-1\}$, $full_{L}\coloneqq{1}$, $lenS\coloneqq{1}$, $lenB\coloneqq$, $ctr\coloneqq{0}$, and change $full_{\ell}$ and $lenS$ correspondingly if the $\ell$th level is non-empty.\label{AlgoRLLine13}
    \end{algorithmic}
\end{algorithm}

\noindent \textbf{Oblivious setup.} In the setup phase, the client utilizes standard cuckoo hashing to store all encrypted elements and secret-shared tags on servers by invoking the algorithm \textbf{Insert}($a,v,L$). During insertion, the servers try to place elements and tags at the $L$th level. If the $L$th level overflows, the elements and tags will be written to the $\ell$th level. Please see the details in Protocol \ref{AlgoInitial}.

\begin{algorithm*}[!t]
    \renewcommand{\thealgorithm}{5}
    \floatname{algorithm}{Protocol}
    \caption{$\Pi_{OpEleRead}${\color{blue}\Comment{Optimized PIR-read for element area.}}}   \label{AlgOpRead}
    
    \begin{multicols}{2}
    \begin{algorithmic}[1]
        \Statex \underline{\textit{Client:}}
        \State \textbf{If} $full_{\ell}\neq{0}$ \textbf{then} access the $\ell$th level as in Protocol \ref{AlgoAccess}. \label{AlgOpReadLine1}
        \State Set $fTab\coloneqq{0},fLen\coloneqq{0},fPos\coloneqq{0}$.{\color{blue}\Comment{$fTab$ indicates whether the required elements are located in the first table ($fTab=0$) or the second table ($fTab=1$) of the cuckoo tables, $fLen$ records the table length, and $fPos$ denotes the position.}}\label{AlgOpReadLine2}
        \State Generate the following two random DPF keys 
        \begin{equation*}
            \begin{aligned}(k_{00},k_{01})\leftarrow{DPF}.Gen(1^{\kappa},rPos_{0},Len_{L},1)\\(k_{10},k_{11})\leftarrow{DPF}.Gen(1^{\kappa},rPos_{1},Len_{L},1)
            \end{aligned}
        \end{equation*}      
        for two random positions $rPos_{0},rPos_{1}$, and then {\color{red}send to two servers} the DPF keys. {\color{blue}\Comment{$Len_{L}$ is the $L$th level's table length.}}\label{AlgOpReadLine3}
        \smallskip
        \Statex \underline{\textit{Servers $\mathcal{S}_{b}$:}}
        \State \textit{Receive} the two DPF keys $k_{0b},k_{1b}$ and then invoke the evaluation function $DPF.Eval(b,k_{0b},j)$ and $DPF.Eval(b,k_{1b},j)$ for each $j\in\{0,...,Len_{L}-1\}$ to generate two vectors $\llbracket{V}_{0}\rrbracket_{b},\llbracket{V}_{1}\rrbracket_{b}$.{\color{blue}\Comment{We omit the subscripts $rPos_{0},rPos_{1}$ for these two vectors.}}\label{AlgOpReadLine4}
        \smallskip
        \Statex \underline{\textit{Client:}}
        \For{level $i=\ell+1,...,L$}\label{AlgOpReadLine5}
            \If{$full_{i}\equiv{0}$} continue.
            \EndIf
            \State \textbf{If} {$found$} \textbf{then}
                \State \quad\quad Generate random positions $posR_{i0},posR_{i1}$ for $i$th level.
            \State \textbf{else then}
                \State \quad\quad Compute the $i$th hash key $hk_{i}\gets{F(lk,i,epoch_{i})}$.
                \State \quad\quad Set $(posR_{i0},posR_{i1})\gets{H(hk_{i},\tau)}$.\label{AlgOpReadLine11}
                
            \State $offset_{0}\coloneqq(rPos_{0}+Len_{i}-posR_{i0})\ mod\ Len_{i}$.
            \State $offset_{1}\coloneqq(rPos_{1}+Len_{i}-posR_{i1})\ mod\ Len_{i}$.\label{AlgOpReadLine12}
            \State {\color{red}Send to two servers} the $offset_{0}$ and $offset_{1}$.\label{AlgOpReadLine13}
            \smallskip
            \Statex \underline{\textit{Servers $\mathcal{S}_{b}$:}}
            \State \textit{Receive} $offset_{0}$ and $offset_{1}$.\label{AlgOpReadLine14}
            \State $\llbracket{V}_{i0}\rrbracket_{b}[k]\coloneqq\oplus_{j=0}^{\frac{Len_{L}}{Len_{i}}-1}{\llbracket{V}_{0}\rrbracket_{b}[j\cdot{Len_{i}+k}]}$ for $k\in\{0,...,Len_{i}-1\}$.\label{AlgOpReadLine15}
            \State $\llbracket{V}_{i1}\rrbracket_{b}[k]\coloneqq\oplus_{j=0}^{\frac{Len_{L}}{Len_{i}}-1}{\llbracket{V}_{1}\rrbracket_{b}[j\cdot{Len_{i}+k}]}$ for $k\in\{0,...,Len_{i}-1\}$.\label{AlgOpReadLine16}
            \State Cyclic left shift $offset_{0}$ for $\llbracket{V}_{i0}\rrbracket_{b}$ and $offset_{1}$ for $\llbracket{V}_{i1}\rrbracket_{b}$.\label{AlgOpReadLine17}
            \State $r_{0b}\coloneqq\oplus_{j=0}^{Len_{i}-1}{\langle{ET_{i0}}\rangle_{b}[j]\cdot{\llbracket{V}_{i0}\rrbracket_{b}[j]}}$.
            \State $r_{1b}\coloneqq\oplus_{j=0}^{Len_{i}-1}{\langle{ET_{i1}}\rangle_{b}[j]\cdot{\llbracket{V}_{i1}\rrbracket_{b}[j]}}$.
            \State {\color{red}Send to client} the $r_{0b}$ and $r_{1b}$.\label{AlgOpReadLine20}
            
            \smallskip
            \Statex \underline{\textit{Client:}}
            \State \textit{Receive} $r_{0b},r_{1b}$ from the server $\mathcal{S}_{b}$.\label{AlgOpReadLine21}
            \State Compute $r_{0}\coloneqq{r_{00}\oplus{r}_{01}}$ and $r_{1}\coloneqq{r_{10}\oplus{r}_{11}}$.
            \State \textit{Decrypt} the elements $r_{0}$ and $r_{1}$.
            
            \If{$r_{0}[0]\equiv{a}$ $\wedge$ ${\neg{found}}$}
                \State $recE\coloneqq{r_{0}}$, $found\coloneqq{true}$, 
                $fTab\coloneqq{0}$,
                $fLen\coloneqq{Len_{i}}$,
                $fPos\coloneqq{posR_{i0}}$, 
            \EndIf
            \If{$r_{1}[0]\equiv{a}$ $\wedge$ ${\neg{found}}$}
                \State $recE\coloneqq{r_{1}}$, $found\coloneqq{true}$, 
                $fTab\coloneqq{1}$,
                $fLen\coloneqq{Len_{i}}$,
                $fPos\coloneqq{posR_{i1}}$,  
            \EndIf
            
        \EndFor
        
        \State \textbf{end for}

        \State \textbf{Output} $fTab,fLen,fPos$. \label{AlgOpReadLine27}

    \end{algorithmic}
   \end{multicols}
\end{algorithm*}

\noindent \textbf{Oblivious access.} Generally, we divide the access phase into three steps: (i) Access each encrypted element from the element buffer and stash in \textit{reverse order} and modify the corresponding tag of the tag buffer and stash using write-only PIR (lines \ref{AlgoAcLine1}-\ref{AlgoAcLine14}); (ii) Utilize read-only PIR to access encrypted elements and use write-only PIR to modify tags sequentially for levels $\{\ell,...,L\}$, and then re-encrypt the accessed element and write it back to the buffer (lines \ref{AlAccETBegin}-\ref{AlAccETEnd}); (iii) Select a level for rebuild based on the access counter $ctr$ (lines \ref{AlAccRebBegin}-\ref{AlAccRebEnd}). Please see Protocol \ref{AlgoAccess} for the detailed description.

\noindent \textbf{Oblivious rebuild\footnote{Note that LO13 \cite{DBLP:conf/tcc/LuO13} and KM19 \cite{DBLP:conf/pkc/KushilevitzM19} require slight adjustments when build the bottom level. We explain the issue in Appendix \ref{AppDiss}.}.} We choose different rebuild protocols for the level $j$: (i) If $j\neq{L}$, the client accesses non-empty elements from tables before $j$th level and remaps the encrypted real and dummy elements to the $j$th level. See the process in Protocol \ref{AlgoRELL}. (ii) If $j\equiv{L}$, the client accesses non-empty elements in all tables, removes the dummies, and only remaps the encrypted real elements to the $L$th level. See the details in Protocol \ref{AlgoRL}.

\smallskip

\noindent\textbf{An example.} Fig. \ref{OurExample} illustrates a four-level example: (1) Before access, the $\ell$th, $\ell+1$th, and $L$th levels contain non-empty elements, while the $\ell+2$th level is empty. (2) During access, the client retrieves all elements from the element buffer $EB$ and stash $ES$, and uses read-only PIR to fetch elements from each table. The client then uses write-only PIR to update the corresponding tags. In this example, the required elements are found at the $\ell+1$th level. The client performs PIR-write on the corresponding position in the tag area of the $\ell+1$th level to modify the tag to a dummy value. For other levels, the client performs pseudo-write operations. (3) Finally, the client rebuilds the $\ell+2$th level. After rebuild, all non-empty tags and elements from the first two levels are stored at the $\ell+2$th level.
\smallskip

\subsection{Overhead Analysis}
\normalsize{We show the theoretical communication cost, bandwidth overhead, and storage cost of Cforam.}

\noindent\textbf{Communication cost (bit).} For the access process:

\begin{itemize}[leftmargin=*]
    \item Firstly, the client obtains elements in $EB$ and $ES$, whose capacities are both $\log{N}$. Then, the client utilizes PIR-write to modify the corresponding tags of the $TB$ and $TS$. The client needs to send two DPF keys and a tag to server $\mathcal{S}_{b}$ for $b\in\{0,1\}$. Thus, the communication cost of this step is $2B\log{N}+4\kappa\log{N}+2\Upsilon$ bits.

    \item Secondly, the client accesses elements for $L-\ell+1$ levels using PIR-read, which transfers four DPF keys and elements for each level. Then, the client modifies the corresponding tags using PIR-write by sending four DPF keys. Thus, this step consumes $(4B+8\kappa\log{N})*(L-\ell+1)$ bits.

    \item Finally, the client sends the accessed element and tag to two servers with $2\cdot(B+\Upsilon)$ bits communication cost.
\end{itemize}

By summing up the above expenses, we calculate the total communication cost for the access phase to be approximately $6B\log{N}+4\kappa\log{N}+8\kappa\log^{2}{N}+4\Upsilon$ bits.

For the rebuild algorithm:

\begin{itemize}[leftmargin=*]
    \item The $\ell$th level contains around $3\log^{2}{N}$ elements and triggers the rebuild protocol during every $\log{N}$ accesses. During the rebuild, the client first requests all elements from one server and tags from two servers. Then, the client sends them to the servers after modification and remapping. Therefore, rebuilding the $\ell$th level consumes ${3\log}^{2}{N}\cdot{(3B+4\Upsilon)}$ bits.

    \item The $i$th level for $i\in\{\ell+1,...,L-1\}$ can accommodate around $2^{i}$ elements, which is rebuilt every $2^{i+1}$ accesses, with a communication cost of ${2}^{i}\cdot{(3B+4\Upsilon)}$ bits.

    \item The $L$th level contains $N$ non-empty elements and executes the rebuild algorithm every $2^{L}$ times of accesses. During the rebuild, firstly, the client receives $2N$ elements and $4N$ tags from servers. Secondly, the client sends $2N$ elements to the server $\mathcal{S}_{0}$, which sends them back to the client after shuffling. Thirdly, after removing the dummies, the client resends the $N$ real elements to the server $\mathcal{S}_{1}$. The server $\mathcal{S}_{1}$ then sends $N$ real elements back to the client after shuffling. Finally, the client invokes \textit{Insert} procedure to reconstruct the $L$th level, which transfers $N$ elements and $2N$ tags. Thus, the $L$th level rebuild process consumes around $2NB+4N\Upsilon+4NB+2NB+NB+2N\Upsilon=9NB+6N\Upsilon$ bits.
\end{itemize}

Therefore, the amortized communication cost of the rebuild is:

\begin{equation}
\begin{aligned}
    & \frac{{3\log}^{2}{N}\cdot{(3B+4\Upsilon)}}{\log{N}}+\sum_{i=\ell+1}^{L-1}{\frac{{2}^{i}\cdot{(3B+4\Upsilon)}}{2^{i+1}}}+\frac{9NB+6N\Upsilon}{2^{L}}\\ & \approx{10B\log{N}+14\Upsilon\log{N}}
\end{aligned}
\end{equation}

By summing up the above cost, Cforam consumes a total of $16B\log{N}+4\kappa\log{N}+8\kappa\log^{2}{N}+14\Upsilon\log{N}$ bits.

Actually, unlike LO13 \cite{DBLP:conf/tcc/LuO13} and KM19 \cite{DBLP:conf/pkc/KushilevitzM19}, Cforam does not require the client to retrieve empty elements in the hash table during the rebuild process. For example, when rebuilding the $\ell$th level, the number of non-empty elements that need to be rebuilt may be less than $3\log^{2}N$. Thus, the actual communication cost is less than the above result.
\smallskip

\noindent\textbf{Bandwidth (\# block).} For a block size $B=\Omega(\kappa\log{N})\gg\Upsilon$, the amortized bandwidth overhead is around $24\log{N}$. For a block size $B=\kappa=\Omega(\log{N})\gg\Upsilon$, the amortized bandwidth becomes $20\log{N}+8\kappa\log{N}$.
\smallskip

\noindent \textbf{Storage.} The client stores two secret keys of $B$ bits for $B = \Omega(\log{N})$, one bit flag $full_{i}$ for each level $i \in \{\ell, \dots, L\}$, two length indicators of $O(\log{N})$ bits each, and a global counter of $O(\log{N})$ bits. Therefore, the local storage is $O(1)$ blocks. The servers store element and tag tables of size $O(N)$.

\section{The Optimized Version Cforam+}\label{SecOptScheme}

The above scheme consumes $O(\kappa\log{N})$ bandwidth for a block size $B=\Omega(\log{N})$. The cost is subject to the second step of the access phase, where the client must generate $O(\kappa\log{N})$-sized DPF keys for each of the $O(\log{N})$ levels. For optimization, we wonder if we can compress the size of DPF keys for less bandwidth consumption.

Hence, we introduce an optimized version, Cforam+, which reduces the bandwidth to $O(\log{N})$ by generating only $O(1)$ DPF keys and some shift information. Cforam+ separates element reading and tag writing of the access phase (modifying lines \ref{AlAccETBegin}-\ref{AlgoAcLine24} of Protocol \ref{AlgoAccess}), while keeping the setup and rebuild phases unchanged. For each access, Cforam+ first reads each level of the element area and outputs auxiliary information that records the location of the required element. This information is used for tag writing. Correctness and security analyses are provided in Appendix \ref{AppCforam+}.

\subsection{Optimized Element Reading}
We present the optimized reading operation for the element area in Protocol \ref{AlgOpRead}. For the $\ell$th level, the access flow remains unchanged, as it does not affect the asymptotic bandwidth complexity (line \ref{AlgOpReadLine1}). For levels $i \in \{\ell, \ldots, L\}$, the client first sends two random DPF keys to each of the two servers, which then execute $DPF.Eval$ to generate two vectors (lines \ref{AlgOpReadLine2}-\ref{AlgOpReadLine4}). Next, for each level $i$, the client sends only $O(\log{N})$-sized shift information to the servers. The servers perform an XOR operation on these vectors to get two new sub-vectors with the length $Len_{i}$, cyclic left shift them, execute the XOR-based inner dot computation, and return the results to the client (lines \ref{AlgOpReadLine5}-\ref{AlgOpReadLine27}). Finally, the client outputs $fTab$, $fLen$, and $fPos$, which are passed as parameters to the optimized writing operation.

\smallskip
\noindent\textbf{An example.} An example of optimized element reading is shown in Fig. \ref{OpreadExample}. Assume that $Len_{L} = 8$ and each server holds a secret-shared unit vector $V_2$. For the $i$th level, where $Len_i = 4$, the server performs the following steps: (1) Splits the secret-shared vector into ${Len_L}/{Len_i} = 2$ sub-vectors; (2) Applies the XOR operation to generate a new vector; (3) Performs a cyclic left shift by an offset of $3$; (4) Calculates the XOR-based inner product between the new vector and the element hash table to produce the final result.

\subsection{Optimized Tag Writing}

Given $fTab$, $fLen$, and $fPos$, we provide the optimized writing operation for the tag area in Protocol \ref{AlgOpWrite}.

Firstly, the operations for the $\ell$th level remain unchanged from Protocol \ref{AlgoAccess} (line \ref{AlgOpWriteLine1}).
Secondly, the client generates a DPF key for the point function $P_{ind,1}$ with a $4{Len}{L}$ input domain, where the target point is $ind = fPos + fLen + fTab \cdot 2{Len}{L}$ (line \ref{AlgOpWriteLine2}). The client then sends the DPF key to two servers.
Thirdly, each server $\mathcal{S}_{b}$ executes $DPF.Eval$ to obtain a vector (line \ref{AlgOpWriteLine3}). For each level $i \in \{\ell+1, \ldots, L\}$, the server divides the vector $V_b$ into two equal sub-vectors and cyclic shifts them by $Len_i$ (lines \ref{AlgOpWriteLine4}-\ref{AlgOpWriteLine6}). Finally, both servers perform PIR-write on the hash tables of the tag area in the range $[0, Len_{i}-1]$ to update tags (lines \ref{AlgOpWriteLine7}-\ref{AlgOpWriteLine8}).

\smallskip
\noindent\textbf{An example.} An example of optimized tag writing is shown in Fig. \ref{OpwriteExample}. Suppose the client finds the desired element at position $1$ in the second hash table of the $i$th level, i.e., $fPos=1,fTab=1,fLen=2$. Initially, the client generates a DPF key for the target point $ind = fPos + fLen + 2Len_L$ and sends it to two servers. Each server then generates a vector with the input domain $4Len_{L}$. Then, it performs the following steps to update tags: (1) Splits the vector into two sub-vectors; (2) Cyclic shifts the second sub-vector by $Len_i$; (3) Updates tags by taking the first $Len_i$ locations of the shifted vector, performing the product with a random tag $\tau_{r}$, and applying XOR operations to modify the tag hash table of the $i$th level.



\begin{figure}[t!]
    \centering
    \subfigure[Element reading]
    {
        \begin{minipage}{\linewidth}
            \centering
            \includegraphics[width=.8\textwidth]{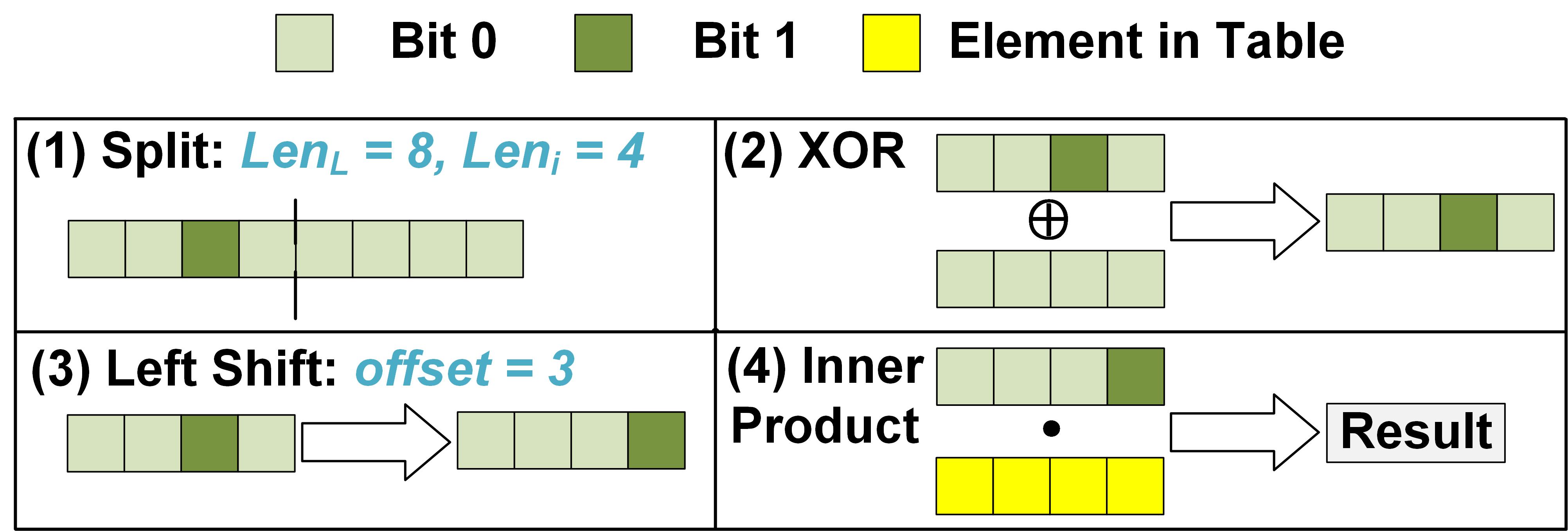}
            \label{OpreadExample}
        \end{minipage}    
    }
    \quad\quad
    \subfigure[Tag writing]
    {
        \begin{minipage}{\linewidth} 
            \centering
            \includegraphics[width=.8\textwidth]{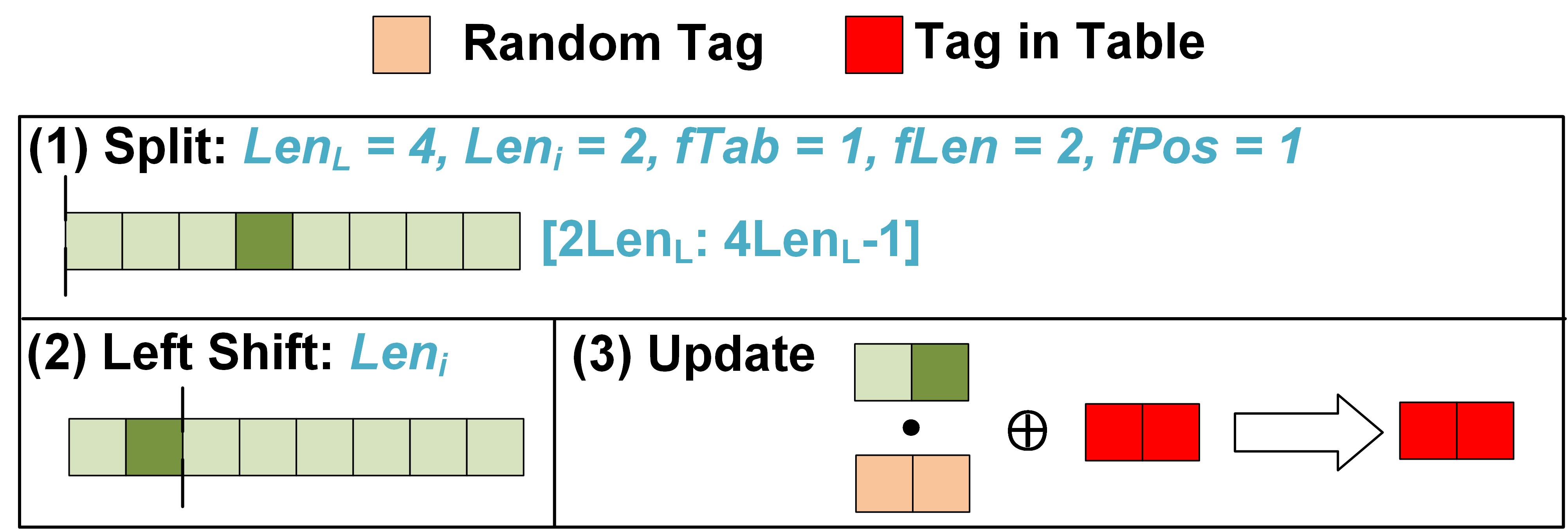}
            \label{OpwriteExample}
        \end{minipage}
    }
    \caption{Examples of optimized reading and writing. For ease of understanding, we only display the flow of the second table at the $i$th level. The diagram uses these colors to represent the components: (i) \textbf{Green}: A light green block represents bit $0$, while a dark green block represents bit $1$. Multiple light green blocks and one dark green block form a secret-shared vector after $DPF.Eval$. (ii) \textbf{Yellow}: Each element in hash table. (iii) \textbf{Red}: Each tag in hash table. (iv) \textbf{Orange}: The random tag $\tau_{r}$ of Protocol \ref{AlgOpWrite}.}
\end{figure}

\subsection{Overhead Analysis}
This section gives the communication cost and bandwidth of Cforam+. Note that the storage cost remains the same as Cforam.

\noindent\textbf{Communication cost (bit).} For the access phase: 

\begin{itemize}[leftmargin=*]
\setlength{\leftmargin}{0pt}

    \item Firstly, the client obtains the elements in $EB$ and $ES$ and modifies the tag in $TB$ and $TS$, which is the same as the Cforam and costs $2B\log{N}+4\kappa\log{N}+2\Upsilon$ bits.
    
    \item Secondly, the client accesses elements and modifies tags for each level. For $\ell$th level, we do the same as Cforam, which consumes $4B+8\kappa\log{N}$ bits. For $i$th level where $i\in\{\ell+1,...,L\}$, the client will invoke Protocol \ref{AlgOpRead} to read elements and then call Protocol \ref{AlgOpWrite} to modify the corresponding tags. In Protocol \ref{AlgOpRead}, the client only sends four DPF keys to servers and then receives $4\cdot(L-\ell)$ results from servers. The communication cost is around $4\cdot(\kappa\log{N}+B\log{N})$ bits. In Protocol \ref{AlgOpWrite}, the client sends a DPF key to two servers with $2\kappa\log{N}$ bits. Thus, this step consumes $14\kappa\log{N}+4B\log{N}$ bits.

    \item Finally, the client sends the accessed element and tag to two servers, which results in $2\cdot(B+\Upsilon)$ bits.
\end{itemize}

\begin{algorithm}[!t]
    \renewcommand{\thealgorithm}{6}
    
    \floatname{algorithm}{Protocol}
    \caption{$\Pi_{OpTagWrite}${\color{blue}\Comment{Optimized PIR-write for tag area.}}} \label{AlgOpWrite}
    
    \begin{algorithmic}[1]
        \Statex \textbf{Input} $fTab,fLen,fPos$.
        \Statex \underline{\textit{Client:}}
        \State \textbf{If} $full_{\ell}\neq{0}$ \textbf{then} access the $\ell$th level as in Protocol \ref{AlgoAccess}.\label{AlgOpWriteLine1}
        \State Compute $(k_{0},k_{1})\leftarrow{DPF}.Gen(1^{\kappa},fPos+fLen+fTab\cdot{2{Len}_{L}},4{Len}_{L},1)$.\label{AlgOpWriteLine2}
        \State {\color{red}Send to server $\mathcal{S}_{b}$} the $k_{b}$ and a random tag $\tau_{r}$.
        \smallskip
        \Statex \underline{\textit{Servers $\mathcal{S}_{b}$:}}
        \State \textit{Receive} $k_{b},\tau_{r}$ and invoke the $DPF.Eval(b,k_{b},j)$ for each $j\in\{0,...,4Len_{L}-1\}$ to generate the vector $\llbracket{V}\rrbracket_{b}$.\label{AlgOpWriteLine3}
        \For{level $i=\ell+1,...,L$}\label{AlgOpWriteLine4}
            \State Split vector $\llbracket{V}\rrbracket_{b}$ into two sub-vectors of equal length denoted as $\llbracket{V_{0}}\rrbracket_{b}\coloneqq{\llbracket{V}\rrbracket_{b}[0:2Len_{L}-1]}$, $\llbracket{V_{1}}\rrbracket_{b}\coloneqq{\llbracket{V}\rrbracket_{b}[2Len_{L}:4Len_{L}-1]}$.\label{AlgOpWriteLine5}
            \State Cyclic left shift $Len_{i}$ for $\llbracket{V_{0}}\rrbracket_{b}$ and $\llbracket{V_{1}}\rrbracket_{b}$.\label{AlgOpWriteLine6}
            \State $\llbracket{TT_{i0}}\rrbracket_{b}[j]\coloneqq{\llbracket{TT_{i0}}\rrbracket_{b}[j]\oplus(\tau_{r}\cdot{\llbracket{V_{0}}\rrbracket_{b}[j]})}$ for $j\in\{0,...,Len_{i}-1\}$.\label{AlgOpWriteLine7}
            \State $\llbracket{TT_{i1}}\rrbracket_{b}[j]\coloneqq{\llbracket{TT_{i1}}\rrbracket_{b}[j]\oplus(\tau_{r}\cdot{\llbracket{V_{1}}\rrbracket_{b}[j]})}$ for $j\in\{0,...,Len_{i}-1\}$.\label{AlgOpWriteLine8}
        \EndFor
        \State \textbf{end for}
    \end{algorithmic}
\end{algorithm}

For the rebuild phase, the cost remains unchanged, i.e., $10B\log{N}+14\Upsilon\log{N}$ bits. By accumulating the above results, our total cost is around $18\kappa\log{N}+16B\log{N}+14\Upsilon\log{N}$ bits.
\smallskip

\noindent\textbf{Bandwidth (\# block).} For the block size $B=\Omega(\kappa\log{N})\gg\Upsilon$, the amortized bandwidth overhead is around $16\log{N}$. While for the block size $B=\kappa=\Omega(\log{N})\gg\Upsilon$, the amortized bandwidth becomes around $34\log{N}$.

\section{Evaluations}\label{SecEval}

To benchmark the performance, we implement our schemes and compare them with LO13 \cite{DBLP:conf/tcc/LuO13}, AFN17 \cite{DBLP:conf/pkc/AbrahamFNP017}, and KM19 \cite{DBLP:conf/pkc/KushilevitzM19}. The code is available at \href{https://github.com/GfKbYu/Cforam}{https://github.com/GfKbYu/Cforam}.

\begin{figure*}[!t]
	\centering

    \centering
    \includegraphics[width=.75\linewidth]{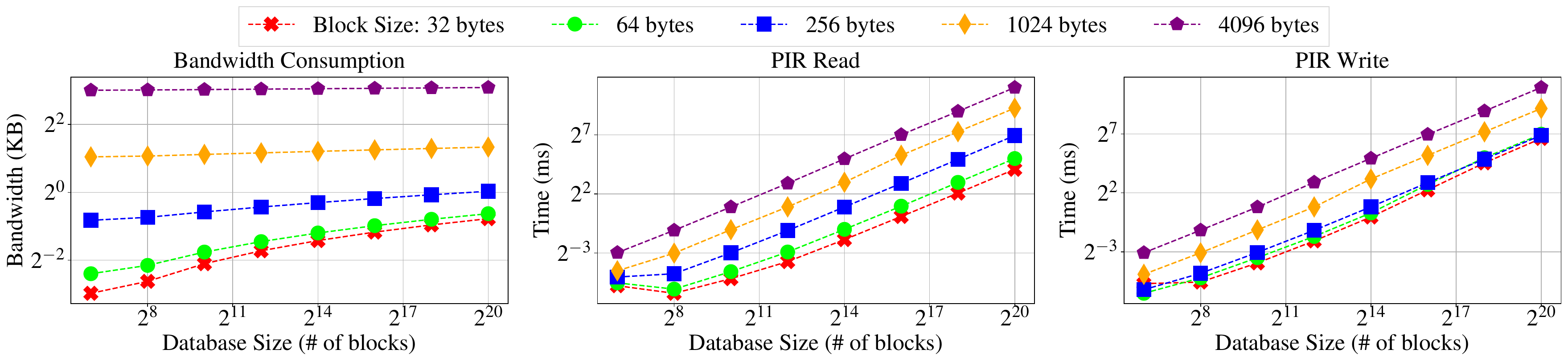}

	\caption{Bandwidth and time costs of a single PIR operation for different database and block sizes. For bandwidth, the consumption of PIR-read and PIR-write operations is the same.}
	\label{fig:pir_cost}
    
\end{figure*}

\subsection{Experimental Setup}

\noindent\textbf{Implementation details.} Our implementations are developed in C++ using gcc version 9.4.0. We use AES-NI instructions to create a high-performance PRF and realize the hash function based on it. Additionally, AVX/AVX2 instructions are leveraged to accelerate PIR computations. For the implementation of LO13 \cite{DBLP:conf/tcc/LuO13}, we set the top level capacity to $c=2\log{N}$ and the boundary level to $\ell_{cuckoo}=\log{\log{N}}$. For levels $i<\ell_{cuckoo}$, we use a standard hash table with $c\cdot2^{i-1}$ buckets, each of size $3\log{N}/\log\log{N}$ as recommended by the authors. For levels $i\geq\ell_{cuckoo}$, we use the two-table cuckoo hashing with $c\cdot2^{i}$ positions.
For AFN17 \cite{DBLP:conf/pkc/AbrahamFNP017}, we deploy the scheme using a $\log{N}$-ary tree and implement the recursive ORAM \cite{DBLP:conf/ccs/StefanovDSFRYD13,garg2016tworam} to store the position map on the servers.
For KM19 \cite{DBLP:conf/pkc/KushilevitzM19}, we configure $d=\log{N}$ so that each level contains $\log{N}-1$ hash tables. We select two-tier hashing \cite{DBLP:conf/asiacrypt/ChanGLS17} to construct hash tables and use the bitonic sort \cite{DBLP:conf/afips/Batcher68} to implement oblivious two-tier hashing.
\smallskip

\noindent \textbf{Platform.} Our benchmark runs on a 64-bit machine with AMD Ryzen 9 5950X 16-Core processors, running on Ubuntu 20.04 x64, with 128 GB RAM and 256 GB SSD.
\smallskip

\noindent \textbf{Network Setup.} All implementations use the TCP protocol to establish network communications. We simulate two network conditions using the linux tool \textbf{tc}: (i) a LAN setting with $80$ Gbit/s bandwidth and $50$ \textmu{s} round-trip time (RTT), and (ii) a WAN setting with $100$ Mbit/s bandwidth and $20$ ms RTT.

\smallskip

\noindent \textbf{Test Metrics.} We begin by measuring the bandwidth and time cost of PIR operations. Then, we evaluate the amortized bandwidth and time cost of ORAM schemes under the following conditions: (i) varying the database size from $2^{6}$ to $2^{20}$ while keeping the block size constant at $32$ bytes, and (ii) varying the block size from $32$ bytes to $4$ kilobytes for a database of size $2^{12}$. In all tests, the number of accesses is set to be an epoch. For hierarchical ORAM schemes, e.g., Ours, LO13, and KM19, an epoch is the number of accesses that trigger the rebuild protocol for the bottom level. For AFN17, an epoch denotes the number of accesses that trigger an eviction.

\subsection{Cost of PIR Operations}
Our schemes rely on DPF-based PIR. Fig. \ref{fig:pir_cost} illustrates the bandwidth and time overhead of PIR operations.

For bandwidth, we observe that it increases with the database size when the block size is relatively small (e.g., $\leq 256$ bytes). However, when the block size is large (e.g., $\geq 1$ KB), bandwidth consumption becomes almost independent of the database size. This behavior is expected since each PIR operation involves transferring two DPF keys of size $\kappa \log{N}$ and two data blocks of size $B$. As the block size approaches $\kappa$, bandwidth is dominated by the DPF keys, leading to a logarithmic increase in bandwidth consumption. When $B \gg \kappa$, the bandwidth stabilizes to the size of two blocks.

Regarding time cost, both PIR-read and PIR-write show linear growth with database size because each PIR operation requires $O(N)$ symmetric key computations. For a database of $2^{20}$ blocks, the time cost of a single PIR computation ranges from $16$ ms to $2$ s. This suggests that PIR computations may become a bottleneck in large-scale databases, especially in a low-latency LAN setting. However, in a WAN setting with a higher latency, the time cost due to PIR computations may not significantly affect the overall ORAM runtime. The next section shows the overall impact of PIR computations on our schemes. Additionally, we discuss computation optimizations through GPU acceleration in Appendix \ref{AppDiss}.

\begin{figure*}[!t]
	\centering

    \centering
    \includegraphics[width=.75\linewidth]{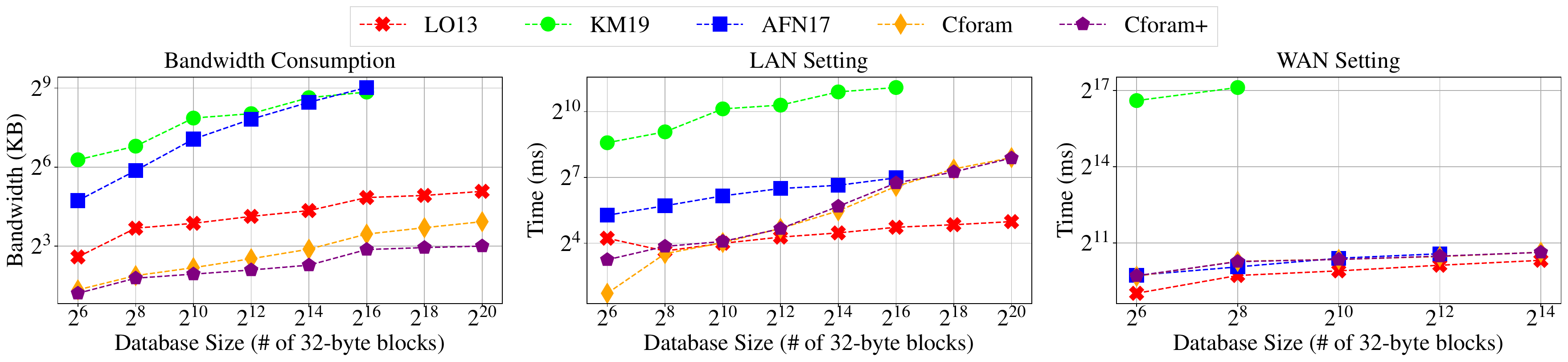}
    \caption{Amortized bandwidth and time cost of our schemes, LO13 \cite{DBLP:conf/tcc/LuO13}, AFN17 \cite{DBLP:conf/pkc/AbrahamFNP017}, and KM19 \cite{DBLP:conf/pkc/KushilevitzM19} for different database sizes.}
	\label{fig:oram_diff_db}
    
\end{figure*}

\begin{figure*}[!t]
	\centering

    \centering
    \includegraphics[width=.75\linewidth]{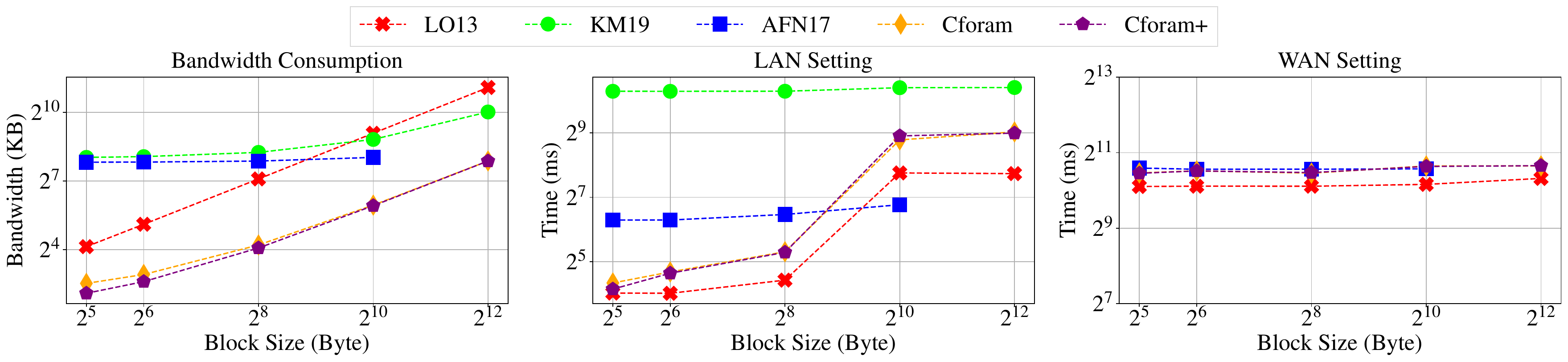}
    \caption{Amortized bandwidth and time cost of our schemes, LO13 \cite{DBLP:conf/tcc/LuO13}, AFN17 \cite{DBLP:conf/pkc/AbrahamFNP017}, and KM19 \cite{DBLP:conf/pkc/KushilevitzM19} for different block sizes.}
	\label{fig:oram_diff_block}
    
\end{figure*}

\subsection{Overall Comparisons}
This section compares our two schemes, Cforam and Cforam+, with LO13 \cite{DBLP:conf/tcc/LuO13}, AFN17 \cite{DBLP:conf/pkc/AbrahamFNP017}, and KM19 \cite{DBLP:conf/pkc/KushilevitzM19} under different database and block sizes. Note that we also test the results of our schemes with the distributed Duoram \cite{DBLP:conf/uss/VadapalliHG23} in Appendix \ref{AppDiss}.

\smallskip

\noindent \textbf{Performance under different database sizes.} Fig. \ref{fig:oram_diff_db} presents the results under different database sizes. Note that AFN17 is tested with database sizes ranging from $2^{6}$ to $2^{16}$, as the server storage exceeds memory size for larger databases. KM19's tests are limited to $2^{16}$ due to its long running time for larger databases.

Cforam+ shows moderate improvement over Cforam. When the database size is $2^{20}$, the bandwidth of Cforam+ is approximately halved compared to Cforam. Both our schemes outperform others in terms of bandwidth. The bandwidth is reduced by $2\sim4\times$ compared to LO13 and by $16\sim64\times$ compared to KM19 and AFN17.

Regarding time cost, our schemes outperform KM19 and AFN17 in the LAN setting. KM19 uses oblivious sort, which involves many interaction roundtrips, resulting in higher time costs. AFN17 requires recursive position maps during the access phase, increasing its runtime. LO13 has a lower time cost than our schemes in the LAN setting, as it does not require $O(N)$ computations. However, in the WAN setting, the time cost due to computations is small. For example, our schemes, LO13, and AFN17 all maintain a time cost of around $2^{10}$ ms, which is significantly better than KM19.

\smallskip

\noindent \textbf{Performance under different block sizes.} Fig. \ref{fig:oram_diff_block} presents the performance of ORAM schemes with different block sizes. We omit the time overhead of KM19 \cite{DBLP:conf/pkc/KushilevitzM19} under the WAN setting, as it requires a huge amortized access time of over $2^{16}$ ms.

As the block size increases, the bandwidth of Cforam gradually approaches that of Cforam+. Compared with Cforam, Cforam+ reduces the number of transmitted DPF keys from $O(\log{N})$ to $O(1)$. For small block sizes (e.g., $32$ bytes), bandwidth is influenced by both the number of transmitted DPF keys and elements, giving Cforam+ a significant advantage. For larger block sizes (e.g., $\geq1$ KB), the number of transmitted elements becomes the primary factor affecting bandwidth, leading to similar bandwidth consumption between our two schemes. Fig. \ref{fig:oram_diff_block} also shows that both our schemes and LO13 \cite{DBLP:conf/tcc/LuO13} exhibit a linear increase in bandwidth with increasing block size, while AFN17 \cite{DBLP:conf/pkc/AbrahamFNP017} and KM19 \cite{DBLP:conf/pkc/KushilevitzM19} show slower growth. This is expected, as AFN17 and KM19 require block sizes of $\Omega(\log^{3}{N})$ and $\Omega(\log^{2}{N})$ respectively to achieve sub-logarithmic bandwidth.

For time cost, our schemes perform similarly to LO13 and AFN17, slightly higher than LO13 and lower than AFN17 for small blocks. Both schemes outperform KM19 in the LAN setting.
In the WAN setting with a high RTT, the performance gap diminishes, as the overall time cost is mainly influenced by transmission time. The results indicate that our schemes, along with LO13 and AFN17, have similar time costs of around $\sim 2^{10}$ ms for different block sizes.

\section{Conclusion}\label{SecConclude}
Since conventional two-server ORAM schemes could be promoted on reducing bandwidth consumption while maintaining $O(1)$ local storage, we propose two new client-friendly schemes based on a pairwise-area setting within a hierarchical structure and DPF-based PIR. These schemes are particularly efficient for lightweight clients with limited resources. Benchmarking results demonstrate that our solutions provide competitive performance in terms of both bandwidth and time cost compared to previous works.

\appendices

\section{Standard Hierarchical ORAM}\label{AppHORAM}

A standard hierarchical ORAM scheme typically contains the setup, the access, and the rebuild phase with the following workflows:

\begin{itemize}[leftmargin=*]
    \item \textbf{Setup phase.} Given $N$ elements, the server constructs $1+\log{N}$ levels, each contains a hash table, denoted as $T_{0},...,T_{\log{N}}$. Each table $T_{i}$ can accommodate up to $2^{i}$ elements, including both real and dummy elements. Initially, the client stores all $N$ real elements in $T_{\log{N}}$.

    \item \textbf{Access phase.} Given the operation $(op,a,writeV)$ where $op\in\{read,write\}$, the client sequentially accesses each non-empty level $T_{0},...,T_{\log{N}}$ based on the hash results of the address $a$. If an element with the address $a$ is found at a position of table $T_{i}$, the client may mark the element at that position as a dummy. Then, the client accesses random positions for the following levels. If $op\equiv{read}$, the client re-encrypts and writes the found element to $T_{0}$; otherwise, the client encrypts and writes the element $(a,writeV)$ to $T_{0}$.

    \item \textbf{Rebuild phase.} After every $2^{i}$ accesses, the client triggers a rebuild of the $i$th level. During the rebuild, the client retrieves and merges all elements from tables $T_{0},...,T_{i}$, and then obliviously stores them at the $i$th level.
\end{itemize}

\section{Two-server PIR Constructions}\label{AppPIR}
Kushilevitz and Mour \cite{DBLP:conf/pkc/KushilevitzM19} introduced constructions for both read-only PIR and write-only PIR based on DPF \cite{DBLP:conf/eurocrypt/GilboaI14,DBLP:conf/eurocrypt/BoyleGI15,DBLP:conf/ccs/BoyleGI16}. We review these two constructions below.

\smallskip
\noindent \textbf{Construction of read-only PIR.}
During the \textbf{setup} phase, given an input array $X$ with length $n$, the client generates two replicas of $X$ denoted as $\langle{X}\rangle_{0}$ and $\langle{X}\rangle_{1}$, and then stores them on two servers $\mathcal{S}_{0}$ and $\mathcal{S}_{1}$, respectively.

To \textbf{read} a value at index $i$, the client interacts with two servers to privately retrieve the value $X[i]$ through the following steps: (i) The client invokes $DPF.Gen(1^{\kappa},i,n,1)$, which outputs the keys $k_{0}$ and $k_{1}$. The client then sends the two keys to two servers, respectively. (ii) Each server $\mathcal{S}_{b}$, where $b\in\{0,1\}$, invokes $DPF.Eval(b,k_{b},j)$ to compute the value $\llbracket{V_{i}}\rrbracket_{b}[j]$ for each $j\in\{0,...,n-1\}$. The server $\mathcal{S}_{b}$ then calculates $r_{b}=\oplus_{j=0}^{n-1}{\langle{X}\rangle_{b}[j]\cdot{\llbracket{V_{i}}\rrbracket_{b}[j]}}$ and sends $r_{b}$ to the client. (iii) The client computes the result $X[i]=r_{0}\oplus{r_{1}}$.
    

The above scheme consumes $O(\kappa{\log{n}}+2B)$ bits and $O(n\log{n})$ symmetric key computations. Note that Boyle et al. \cite{DBLP:conf/ccs/BoyleGI16} introduced an optimization technique called full domain evaluation, which allows the evaluation of all points with  $2n$ computational overhead.

\smallskip

\noindent \textbf{Construction of write-only PIR.} Similarly, DPF can be utilized to implement write-only PIR. \textbf{Initially}, the client uploads the secret sharing of $X$, denoted as $\llbracket{X}\rrbracket_{0}$ and $\llbracket{X}\rrbracket_{1}$, to two servers.

When the client wishes to \textbf{write} a new value $x_{new}$ at position $i$,  assuming the client knows the old value  $x_{old}$ currently stored at position $i$ of the array $X$, the client calculates $x_{mod}=x_{new}\oplus{x_{old}}$. Afterwards, the client interacts with two servers to write the value using the following steps: (i) The client invokes $DPF.Gen(1^{\kappa},i,n,1)$, which outputs the keys $k_{0}$ and $k_{1}$. The client then sends the two keys and the value $x_{mod}$ to two servers. (ii) Each server $\mathcal{S}_{b}$, where $b\in\{0,1\}$, invokes $DPF.Eval(b,k_{b},j)$ to calculate the $\llbracket{V_{i}}\rrbracket_{b}[j]$ and overwrites values at position $j$ with $\llbracket{X}\rrbracket_{b}[j]\oplus({{x}_{mod}\cdot{\llbracket{V_{i}}\rrbracket_{b}[j]}})$ for each $j\in\{0,\ldots,n-1\}$, where $\llbracket{X}\rrbracket_{b}[j]$ is the value at position $j$ before the write operation.

We can implement the function $\mathbf{Build}$ by simply computing $X'=\llbracket{X}\rrbracket_{0}\oplus{\llbracket{X}\rrbracket_{1}}$, where $\llbracket{X}\rrbracket_{0}$ and $\llbracket{X}\rrbracket_{1}$ are the arrays currently stored on two servers, respectively.

The above scheme also consumes $O(\kappa{\log{n}}+2B) $ bits and $O(n)$ computational cost, similar to read-only PIR.

The correctness and security are guaranteed by the DPF. Readers can follow \cite{DBLP:conf/pkc/KushilevitzM19} for the analysis.

\section{Correctness and Security of Cforam}\label{AppCforam}

\subsection{Correctness}
The following three lemmas guarantee the correctness of Cforam.

\begin{lemma}
\normalfont
    If an element is accessed for the first time, it will be found at the $L$th level or $\ell$th level.
\end{lemma}
\begin{proof}
\normalfont
    In the setup phase, the client attempts to store all elements at the $L$th level. Up to $O(\log{N})$ elements may overflow and be placed at the $\ell$th level. The $\ell$th level can accommodate $O(\log^2{N})$ elements, with an $O(\log{N})$-sized stash, which satisfies the conditions of Theorem 1 and has a negligible build failure probability. Therefore, the first access to an element will result in it being found at either the $L$th level or the $\ell$th level.
\end{proof}

\begin{lemma}
\normalfont
     Two real elements with a same virtual address `$a$' will never appear simultaneously at the same level.
\end{lemma}
\begin{proof}
\normalfont
    In the access phase, the client writes the accessed element back to the element buffer $EB$ and assigns it a new random tag. In the rebuild phase, the client updates the elements according to their tags. Specifically, if the computed tag matches the accessed tag, the client writes the element back to the servers without modification. Otherwise, the element is transformed into a dummy and written back to the servers. This ensures that the real element corresponding to an access will always be replaced by a dummy during the rebuild phase. As a result, real elements with the same virtual address will not appear simultaneously at the same level. 
\end{proof}

\begin{lemma}
\normalfont
    The freshest elements must appear in the element buffer $EB$ in reverse order, or in the hash tables $ET_{\ell}$, \dots, $ET_L$ of each level, in sequence. In other words, the first accessed element matching the query must be the latest one written in ORAM.
\end{lemma}
\begin{proof}
\normalfont
    For each access request, the accessed element is added into an empty location in the element buffer $EB$ in sequence. As elements are added, those with larger indexes are fresher. Once the number of elements in $EB$ reaches $\log{N}$, the protocol triggers the rebuild of the $\ell$th level, transferring all fresh elements from $EB$ to the $\ell$th level. Once the $\ell$th level is full, the client continues by writing the elements to the next level, and so on. Consequently, new elements from smaller levels are always rebuilt into the next empty level. Therefore, the freshest elements must appear in $EB$ in reverse order, or in the hash tables $ET_{\ell}$, \dots, $ET_L$ of each level, in sequence, as the levels are rebuilt.  
\end{proof}

\subsection{Security}

We claim that Cforam is secure based on the pseudorandomness of PRF, the security of PIR, and the negligible build failure probability of cuckoo hashing. We define a sequence of hybrid games to show that the real world and the ideal world are computationally indistinguishable. 

In Cforam, (i) symmetric encryption ensures the indistinguishability of elements; (ii) all tags are stored on two non-colluding servers in a secret-shared manner and cannot be recovered by one of these two servers. As a result, the transmitted and stored elements and tags (i.e., $Out_{b}$ and $EX_{b}$) do not reveal any additional information to the adversary. Therefore, we do not further discuss the security of $Out_{b}$ and $EX_{b}$ in the following proof. Instead, we focus on proving the indistinguishability of the access pattern $AP_{b}$.

\smallskip

\noindent \textbf{Game 1 (real world).} This is the game $\mathbf{Real}_{\mathcal{A},\mathcal{C}}^{\Pi_{ORAM}}$ in the real world, where we implement the setup protocol using Protocol \ref{AlgoInitial} and the access protocol using Protocol \ref{AlgoAccess}. For each operation, the corresponding real protocol is executed, and the adversary receives the output of the execution.
\smallskip

\noindent \textbf{Game 2.} Game 2 differs from Game 1 only during the access phase (i.e., Protocol \ref{AlgoAccess}). For each access, the simulator first retrieves all elements from the element buffer $EB$ and element stash $ES$ without recording any additional information (corresponding to lines \ref{AlgoAcLine1}-\ref{AlgoAcLine12}). The simulator then performs pseudo-write operations using the write-only PIR, writing a random tag at position $0$ in the tag buffer $TB$ and stash $TS$ (replace lines \ref{AlgoAcLine13}-\ref{AlgoAcLine14}). Next, the simulator requests random positions from non-empty levels of the element area using the read-only PIR (replace lines \ref{AlgoAcLine15}-\ref{AlgoAcLine23}) and performs pseudo-writes in the tag area using the write-only PIR (replace line \ref{AlgoAcLine24}). Finally, the simulator sends an encrypted dummy element and a random secret-shared tag to two servers (replace line \ref{AlgoAcLine28}).

\begin{theorem}
\normalfont
    Game 2 is computationally indistinguishable from Game 1 based on the security of the read-only PIR and the write-only PIR.
\end{theorem}
\begin{proof}
\normalfont
    In Game 2, the simulator performs the same operations as in Game 1 for accessing the element buffer $EB$ and element stash $ES$. The difference between Game 1 and Game 2 begins when the simulator proceeds to update the tag buffer $TB$ and stash $TS$. In Game 2, the simulator uses write-only PIR to perform a pseudo-write operation instead of writing to specific locations as in Game 1. Then, instead of accessing specific locations of each non-empty level, the simulator retrieves random positions from the element area and updates position $0$ of the tag area using two-server PIR. Due to the security of PIR, the adversary cannot distinguish between the actual read/write locations and the randomly chosen ones. Thus, Game 2 is computationally indistinguishable from Game 1.
\end{proof}

\noindent \textbf{Game 3.} In this game, only the rebuild phase is modified, while all other processes remain the same as in Game 2. (i) When rebuilding the $j$th level where $j \neq L$ (i.e., Protocol \ref{AlgoRELL}), the simulator sequentially accesses non-empty elements from the previous $j-1$ levels, one by one (same as lines \ref{AlRebJ1}-\ref{AlRebJ8}). For each non-empty element, the simulator invokes the \textbf{Insert} algorithm to place the element at the $j$th level. However, unlike in Game 2, where positions are determined using PRF-based hash functions, in Game 3 the simulator randomly selects the positions and sends the encrypted element, tag, and the randomly selected positions to the servers. The two servers then place the element and tag in the corresponding positions according to the cuckoo strategy. (ii) When rebuilding the $L$th level (i.e., Protocol \ref{AlgoRL}), the simulator first receives all non-empty elements and tags from server $\mathcal{S}_{0}$ and non-empty tags from server $\mathcal{S}_{1}$ (same as lines \ref{AlgoRLLine1}-\ref{AlgoRLLine2}). Then, the simulator re-encrypts each element and sends it to $\mathcal{S}_{0}$ without modifying it (replace lines \ref{AlgoRLLine3}-\ref{AlgoRLLine4}). Next, the simulator receives the shuffled elements from $\mathcal{S}_{0}$, removes dummy elements, and sends the remaining $N$ real elements to $\mathcal{S}_{1}$ (same as lines \ref{AlgoRLLine5}-\ref{AlgoRLLine8}). Finally, the simulator receives each shuffled element from $\mathcal{S}{1}$ and invokes the \textbf{Insert} algorithm to reconstruct the $L$th level. Similarly to the previous step, the simulator randomly generates positions for each element instead of using hash functions.

\begin{theorem}
\normalfont
    Game 3 is computationally indistinguishable from Game 2 based on the pseudorandomness of PRF and the negligible build failure probability of cuckoo hashing.
\end{theorem}
\begin{proof}
\normalfont
    (i) When rebuilding the $j$th level where $j \neq L$, the simulator sequentially accesses each non-empty element as in Game 2. The key difference is that, instead of using PRF-based hash functions to compute positions, the simulator randomly selects positions and sends these positions along with the elements to the servers. Since PRF is computationally indistinguishable from a truly random function, the adversary cannot distinguish between these two methods based on location information. Additionally, during the rebuild, up to $O(\log{N})$ elements may overflow and be placed in the $\ell$th level. The $\ell$th level contains an $O(\log{N})$ stash, satisfying the condition of Theorem 1 and having a negligible build failure probability. Thus, the adversary cannot distinguish between Game 2 and Game 3 during the rebuild of levels $j \neq L$. (ii)
    When rebuilding the $L$th level, two servers send the elements and tags to the simulator as in Protocol \ref{AlgoRL}. The simulator then sends elements to server $\mathcal{S}_0$, receives shuffled elements, removes dummies, and re-sends the remaining real elements to server $\mathcal{S}_1$. The simulator then invokes the \textbf{Insert} algorithm to reconstruct the $L$th level. Since the elements are shuffled between the servers before reconstruction, a single server cannot distinguish between elements before and after removal. Moreover, during the insertion process, the simulator sends random positions to each server. Using the same reasoning as in (i), we can conclude that Game 2 and Game 3 are indistinguishable during the rebuild of the $L$th level. Therefore, Game 3 is computationally indistinguishable from Game 2.

\end{proof}

\noindent \textbf{Game 4 (ideal world).} Game 4 is similar to Game 3, except during the execution of the Setup, where the following changes occur: (i) the simulator randomly generates $|X_{b}|$ random real elements for encryption, instead of using the elements from the array $X_{b}$; (ii) the simulator selects random positions for each element for insertion, rather than those determined by the hash functions.

\begin{theorem}
\normalfont
    Game 4 is computationally indistinguishable from Game 3 based on the pseudorandomness of PRF and the negligible build failure probability of cuckoo hashing.
\end{theorem}
\begin{proof}
\normalfont
    Replacing the real elements of the database with random elements does not affect the overall distribution of the database, as the symmetric encryption can preserve the indistinguishability of the elements. Additionally, the positions of the elements are now chosen randomly instead of being derived from the PRF-based hash functions. Due to the pseudorandomness of PRF, the adversary cannot distinguish between these two games based on position information. Finally, the build failure probability is negligible based on Theorem 1. Thus, Game 4 $\equiv$ Game 3.
\end{proof}

In Game 4, the simulator can simulate the whole protocol with knowledge of $|X_{b}|$ only\footnote{We have not explicitly mentioned other public parameters, such as the security parameter $\lambda$, the size of the cuckoo hashing table, block size, etc.}. In other words, Game 4 simulates the game $\mathbf{Ideal}_{\mathcal{A},\mathcal{C}}^{Sim}$ of the ideal world. Therefore, we can conclude that Game 1 (Real) $\equiv$ Game 2 $\equiv$ Game 3 $\equiv$ Game 4 (Ideal) and complete the proof.


\section{Correctness and Security of Cforam+}\label{AppCforam+}

\subsection{Correctness} 
Since we only adjust the element reading and tag writing of the $i$th level for $i\in\{\ell+1,...,L\}$, we need to prove that the optimized reading and writing maintain the same functionality as Cforam. We discuss the following three situations.
\smallskip

\noindent\textbf{Situation 1.} Assume that the required element is stored in the EB, ES, or the $\ell$th level. In Cforam+, first, the client retrieves the required elements and updates the corresponding tags in TB, TS, and $TT_{\ell}$, just like in Cforam. Second, the client performs the optimized \textbf{element reading} operations for each level $i \in \{\ell + 1, \dots, L\}$, obtaining parameters $fTab = 0$, $fPos = 0$, and $fLen = 0$. Third, during the \textbf{tag writing} phase, each server receives a $0$th unit vector $V_{0}$ after the DPF evaluation. The server then splits the vector into two equal-length sub-vectors. For the second table of $i$th level, the tag remains unchanged since all values in the second sub-vector are always secret-shared $0$. For the first table of $i$th level, the client performs a cyclic left shift operation to obtain $V_{ind}$, where:
\begin{equation}
    \begin{aligned}
        ind &= (0-Len_{i})\ mod\ 2Len_{L} \\
        &\geq{Len_{L}}.
    \end{aligned}
\end{equation}
We observe that the tag also remains unchanged, as all values in the vector $V[:Len_{i}-1]$ are secret-shared $0$.

We conclude that in Cforam+, the client can retrieve the required elements from EB, ES, or the $\ell$th level without altering the tag of the $i$th level for $i \in \{\ell + 1, \dots, L\}$. This ensures that Cforam+ achieves the same functionality as Cforam.

\smallskip

\noindent\textbf{Situation 2.} Assume that the required element is stored at position $posR_{i1}$ in the \textit{second} table of the $i$th level, where $i \in \{\ell + 1, \dots, L\}$, i.e., $fTab = 1$, $fPos = posR_{i1}$, and $fLen = Len_{i}$.

For \textbf{element reading}, in Cforam, the client generates a unit vector $V_{posR_{i1}}$ and sends the secret-shared version of $V_{posR_{i1}}$ to two servers. The two servers perform DPF evaluation and XOR-based inner product operations based on $V_{posR_{i1}}$. In Cforam+, the client sends the secret-shared unit vector $V_{rPos_{1}}$ to two servers given a random position $rPos_{1}$. Each server performs XOR computations to access $V_{rPos_{1} \mod Len_{i}}$. Subsequently, the two servers cyclic left-shift the vector by $offset_{1}$ to access a new vector $V_{ind}$ such that:

\begin{equation}
    \begin{aligned}
        ind &= \big((rPos_{1}\ mod\ Len_{i})-offset_{1}\big)\ mod\ Len_{i} \\
        &= (posR_{i1}-Len_{i})\ mod\ Len_{i} \\
        &= posR_{i1}.
    \end{aligned}
\end{equation}

Thus, the servers obtain the same vector and do the same computations as in Cforam.

For \textbf{tag writing}, the servers get the following vector after shifting for each level:

\begin{itemize}[leftmargin=*]
    \item For the second table of the $i$th level, the two servers obtain the vector $V_{ind}$ such that:
    \begin{equation}
        \begin{aligned}
            ind &= (fPos+fLen-Len_{i})\ mod\ 2Len_{L} \\
            &= posR_{i1}.
        \end{aligned}
    \end{equation}

    \item For the second table of the $l$th level where $l<i$, the two servers obtain the vector $V_{ind}$ such that:
    \begin{equation}
        \begin{aligned}
            ind &= (fPos+fLen-Len_{l})\ mod\ 2Len_{L} \\
            &\geq{(fPos+2Len_{l}-Len_{l})\ mod\ 2Len_{L}} \\
            &\geq{Len_{l}}.
        \end{aligned}
    \end{equation}

    \item For the second table of the $l$th level where $l>i$, the two servers obtain the vector $V_{ind}$ such that:
    \begin{equation}
        \begin{aligned}
            ind &= (fPos+fLen-Len_{l})\ mod\ 2Len_{L} \\
            &\geq{(fPos+Len_{i}+Len_{L})\ mod\ 2Len_{L}} \\
            &\geq{Len_{L}}.
        \end{aligned}
    \end{equation}

    \item For the first table of each level, each value in vector $V$ is always the secret-shared $0$.
\end{itemize}

To summarize the above results, we observe that the servers always correctly modify the tag of the second table at the $i$th level. For the first table at the $i$th level, the tag remains unchanged, as each value in the vector $V$ is the secret-shared $0$. For all other levels where $l \neq i$, the tag also remains unchanged, since each value in the vector $V[:Len_{l}-1]$ is the secret-shared $0$.

\smallskip

\noindent\textbf{Situation 3.} Assume that the required element is stored at the $i$th level's \textit{first} table where $i\in\{\ell+1,...,L\}$. This is similar to situation 2, and following its process can easily prove the correctness of situation 3.

\subsection{Security}

In this section, we demonstrate that the adversary cannot infer the position of the required element from the optimized element reading and tag writing processes. Note that the other parts remain consistent with Cforam.
\smallskip

\noindent\textbf{Element Reading.} If the required element is stored at the $\ell$th level, the offsets sent to the servers are always random according to our protocol. Now, assume the required element is stored at the $i$th level, where $i \in \{\ell+1, \dots, L\}$. For element reading, the client accesses each level sequentially as follows:

\noindent
\begin{itemize}[leftmargin=*]
    \item When the client requests a level $l < i$, (s)he sends two offsets to the servers based on the positions after computing hash functions. Since the $l$th level does not store the required element, these positions and offsets appear random to the servers due to the pseudorandomness of PRF.

    \item When the client requests a level $l > i$, (s)he sends two random offsets to the servers according to our protocol.

    \item When the client requests the $i$th level, (s)he sends two offsets based on cuckoo hashing. If the client were to send the position information directly to the servers, they could link the position in the access phase with that in the rebuild phase. This would occur because our rebuild algorithm does not obscure which locations store non-empty elements. To preserve security, the client generates random DPF keys for access, ensuring the offsets remain random. 
\end{itemize}

In summary, the offsets at each level appear random to the servers, preventing them from discerning the specific level and position of the required elements. Furthermore, in Cforam+, the same element at a level is never accessed twice before reconstruction. As a result, the offsets remain random for any number of element reads, making it impossible for the servers to identify the required elements.


\smallskip

\noindent\textbf{Tag Writing.} For tag writing, each server only receives a DPF key to compute a secret-shared vector and performs a cyclic shift by a fixed distance for each level. Consequently, the servers cannot distinguish the location at which the tag is modified.

\section{Supplementary Work}\label{AppDiss}

\noindent\textbf{Applications of ORAM with low local storage.} Client storage is often limited in many real-world environments, such as on mobile devices and embedded systems. Furthermore, numerous studies \cite{DBLP:conf/ndss/SasyGF18, mishra2018oblix, ahmad2019obfuscuro, reichert2024menhir} have explored extending ORAM to resource-constrained Trusted Execution Environments (TEEs) (e.g., Intel SGX's EPC is limited to 128MB) to minimize local storage overhead. In this scenario, the TEE enclave enables all client computation, storage, and access tasks. However, unlike the traditional client-server ORAM model, where all client-side routines are executed in a fully trusted environment, the adversary in a TEE-based ORAM can observe the page-level access addresses of the enclave \cite{DBLP:conf/sp/XuCP15, DBLP:conf/uss/BulckWKPS17}. This visibility may compromise the security of the ORAM scheme. Therefore, to securely extend ORAM to TEEs, any access involving the TEE’s internal memory must rely on techniques such as linear scans or other oblivious algorithms, which can increase both access latency and computational costs. These additional expenses can be eliminated for the ORAM with constant local storage since linearly scanning the internal memory does not increase the asymptotic complexity. Therefore, it is significantly applicable to construct Client-friendly ORAM with low local storage, especially constant local storage with practical bandwidth.

\smallskip

\noindent\textbf{Comparisons with distributed ORAM.} In standard multi-server ORAM, the client interacts with each server \textit{independently}, with no communication between the servers. A variant of multi-server ORAM is distributed ORAM (DORAM), where the client sends a secret-shared query to multiple servers, then the servers communicate with each other before responding to the client. Many studies have been conducted on DORAM, addressing various aspects such as achieving malicious security \cite{DBLP:conf/ndss/HoangGY20, DBLP:conf/tcc/HemenwayNOSZ23}, optimizing bandwidth \cite{DBLP:conf/tcc/LuO13, DBLP:conf/scn/FalkNO22, DBLP:conf/uss/VadapalliHG23}, and reducing computation costs \cite{DBLP:conf/pkc/HamlinV21, braun2023ramen}. Our research primarily focuses on the standard two-server ORAM model, where the client interacts with each server separately.

We conducted experiments comparing our two-server ORAM with the concretely efficient two-party DORAM solution, Duoram \cite{DBLP:conf/uss/VadapalliHG23}, as shown in Fig. \ref{fig:duoram}. Experimental results indicate that, compared to Duoram, we achieve approximately $90\%$ savings in bandwidth and reduce the time cost by a factor of $2\sim4\times$. This suggests that, in terms of bandwidth and runtime efficiency, the standard client-server ORAM still holds advantages over existing DORAM solutions. 

It is worth noting that our constructions can also be extended to DORAM scenarios by implementing the client within a secure computation framework \cite{OS97, DBLP:conf/tcc/LuO13}. However, this approach introduces an additional $O(\lambda)$ overhead due to the need for secure computation to evaluate the PRF function. Therefore, an interesting question for future work is whether we can extend our constructions to DORAM without relying on black-box secure computation for evaluating the PRF, while still maintaining competitive performance.

\begin{figure}[!t]
	\centering

    \centering
    \includegraphics[width=0.7\linewidth]{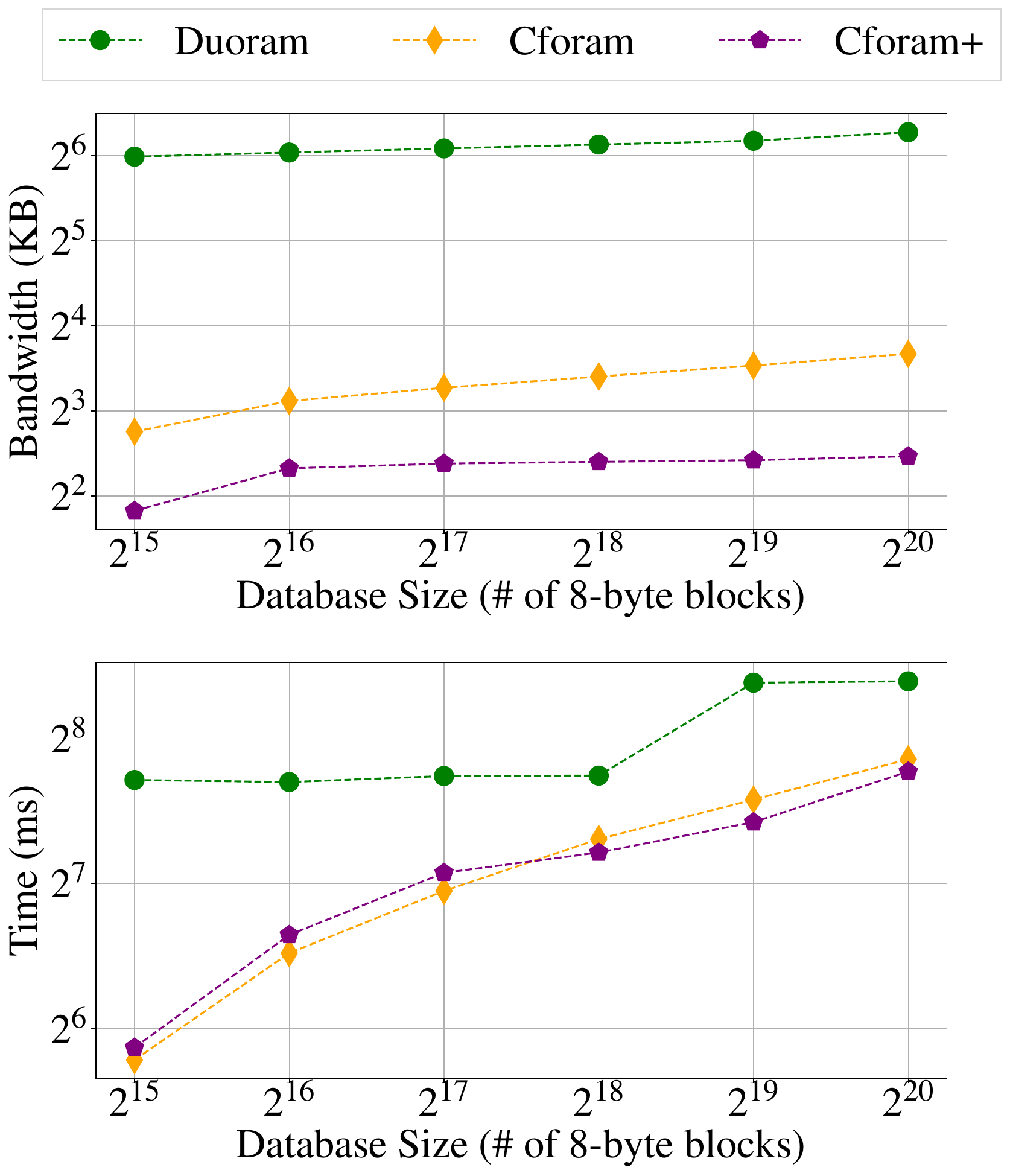}

	\caption{Amortized bandwidth and time cost in the LAN setting compared with Duoram \cite{DBLP:conf/uss/VadapalliHG23}.}
	\label{fig:duoram}
\end{figure}


\smallskip

\noindent\textbf{Setup and rebuild of other schemes.} LO13 \cite{DBLP:conf/tcc/LuO13}, AFN17 \cite{DBLP:conf/pkc/AbrahamFNP017}, and KM19 \cite{DBLP:conf/pkc/KushilevitzM19} do not explicitly describe the setup phase. We provide the following supplements for these schemes:

\begin{itemize}[leftmargin=*]
    \item For LO13 (hierarchical ORAM), we use $\mathcal{S}_{b}$ to represent the server that stores the bottom level $L$. In the setup phase, the client first sends all elements to $\mathcal{S}_{b\oplus{1}}$ in sequence. Then, the server $\mathcal{S}_{b\oplus{1}}$ constructs a hash table to store these elements. Next, the server $\mathcal{S}_{b\oplus{1}}$ sends the hash table to the client, who then forwards it to $\mathcal{S}_{b}$. Finally, the server $\mathcal{S}_{b\oplus{1}}$ deletes the $L$th table.

    \item For AFN17 (tree-based ORAM), during the setup phase, the client performs the following steps in a loop until all elements are placed in the tree: (i) sends $d$ elements to two servers, which store them in the root bucket; (ii) invokes the eviction algorithm after every $d$ elements.

    \item For KM19 (hierarchical ORAM), two servers initially establish the largest level $L$ with a capacity of $O(N)$. Then, the client sends all the elements to the two servers and uses two-tier oblivious hashing \cite{DBLP:conf/asiacrypt/ChanGLS17} to construct the $L$th level.
\end{itemize}

In LO13 and KM19, when the number of accesses equals an epoch, their rebuild protocols will no longer apply to the bottom level because the number of non-empty elements exceeds the ORAM capacity\footnote{Note that LO13 proposed adding one extra level at the bottom for every $O(N)$ access. We show that this is unnecessary and may lead to an infinite increase in storage.}. To address this, we make the following subtle adjustments:

\begin{itemize}[leftmargin=*]
    \item For LO13, assuming that server $\mathcal{S}_{b}$ stores the bottom level, the following process is required: (i) Server $\mathcal{S}_{b\oplus{1}}$ sends all its shuffled elements to the client, who then forwards them to server $\mathcal{S}{b}$; (ii) Server $\mathcal{S}{b}$ combines the received elements with all its own, shuffles the resulting array, and sends it to the client; (iii) The client removes \textit{empty elements} and \textit{dummy elements} from the received array, and only sends the real elements to server $\mathcal{S}{b \oplus 1}$; (iv) Server $\mathcal{S}{b \oplus 1}$ constructs a hash table to store the real elements and sends it to server $\mathcal{S}_{b}$ through the client.

    \item KM19 uses oblivious sorting to construct oblivious hashing. We only need to adjust the oblivious sort algorithm: For non-bottom levels, the sorting algorithm must place real and dummy elements at the front of the array and write them into hash tables. For the bottom level, we adjust the comparison conditions so that dummy elements are always placed after the real elements in order to remove the dummy elements. 
\end{itemize}
\smallskip

\noindent\textbf{Time cost of our schemes.} From the experimental results, we observe that in the LAN setting, the running time of our schemes does not show significant advantages over LO13 \cite{DBLP:conf/tcc/LuO13} and AFN17 \cite{DBLP:conf/pkc/AbrahamFNP017}. This is primarily due to the linear symmetric key computations involved in PIR. Recent work \cite{dan2020DPF, lam2023gpu} has explored how multi-threaded CPUs and GPUs can be leveraged to accelerate DPF-based PIR computations. For instance, results in \cite{lam2023gpu} show that for a database size of $2^{20}$ blocks, a PIR-read operation can be completed in just $1$ ms using C++ compilation and GPU acceleration. This is a positive outcome for the deployment of our schemes. For further details, we refer the reader to \cite{lam2023gpu}.

\bibliographystyle{IEEEtran}
\bibliography{TIFS_ref}

\begin{thebibliography}{10}
\providecommand{\url}[1]{#1}
\csname url@samestyle\endcsname
\providecommand{\newblock}{\relax}
\providecommand{\bibinfo}[2]{#2}
\providecommand{\BIBentrySTDinterwordspacing}{\spaceskip=0pt\relax}
\providecommand{\BIBentryALTinterwordstretchfactor}{4}
\providecommand{\BIBentryALTinterwordspacing}{\spaceskip=\fontdimen2\font plus
\BIBentryALTinterwordstretchfactor\fontdimen3\font minus \fontdimen4\font\relax}
\providecommand{\BIBforeignlanguage}[2]{{%
\expandafter\ifx\csname l@#1\endcsname\relax
\typeout{** WARNING: IEEEtran.bst: No hyphenation pattern has been}%
\typeout{** loaded for the language `#1'. Using the pattern for}%
\typeout{** the default language instead.}%
\else
\language=\csname l@#1\endcsname
\fi
#2}}
\providecommand{\BIBdecl}{\relax}
\BIBdecl

\bibitem{AWE}
Microsoft, ``Always encrypted sql server,'' \url{https://learn.microsoft.com/en-us/sql/relational-databases/security/encryption/always-encrypted-database-engine?view=sql-server-ver16}, 2024.

\bibitem{MongoDB}
M.~Inc, ``Mongodb,'' \url{https://www.mongodb.com/}, 2023.

\bibitem{DBLP:conf/ndss/IslamKK12}
M.~S. Islam, M.~Kuzu, and M.~Kantarcioglu, ``Access pattern disclosure on searchable encryption: Ramification, attack and mitigation,'' in \emph{NDSS}, 2012.

\bibitem{DBLP:conf/ccs/CashGPR15}
D.~Cash, P.~Grubbs, J.~Perry, and T.~Ristenpart, ``Leakage-abuse attacks against searchable encryption,'' in \emph{CCS}, 2015, pp. 668--679.

\bibitem{DBLP:conf/ccs/KellarisKNO16}
G.~Kellaris, G.~Kollios, K.~Nissim, and A.~O'Neill, ``Generic attacks on secure outsourced databases,'' in \emph{CCS}, 2016, pp. 1329--1340.

\bibitem{DBLP:conf/asiacrypt/GordonKW18}
S.~D. Gordon, J.~Katz, and X.~Wang, ``Simple and efficient two-server {ORAM},'' in \emph{ASIACRYPT}, 2018, pp. 141--157.

\bibitem{DBLP:conf/tcc/LuO13}
S.~Lu and R.~Ostrovsky, ``Distributed oblivious {RAM} for secure two-party computation,'' in \emph{TCC}, 2013, pp. 377--396.

\bibitem{DBLP:conf/pkc/AbrahamFNP017}
I.~Abraham, C.~W. Fletcher, K.~Nayak, B.~Pinkas, and L.~Ren, ``Asymptotically tight bounds for composing {ORAM} with {PIR},'' in \emph{PKC}, 2017, pp. 197--214.

\bibitem{DBLP:conf/pkc/KushilevitzM19}
E.~Kushilevitz and T.~Mour, ``Sub-logarithmic distributed oblivious {RAM} with small block size,'' in \emph{PKC}, 2019, pp. 3--33.

\bibitem{DBLP:conf/stoc/Goodrich14}
M.~T. Goodrich, ``Zig-zag sort: a simple deterministic data-oblivious sorting algorithm running in o(n log n) time,'' in \emph{STOC}, 2014, pp. 684--693.

\bibitem{DBLP:conf/afips/Batcher68}
K.~E. Batcher, ``Sorting networks and their applications,'' in \emph{American Federation of Information Processing Societies, AFIPS}, 1968, pp. 307--314.

\bibitem{DBLP:journals/jacm/GoldreichO96}
O.~Goldreich and R.~Ostrovsky, ``Software protection and simulation on oblivious rams,'' \emph{J. {ACM}}, vol.~43, no.~3, pp. 431--473, 1996.

\bibitem{DBLP:conf/tcc/WeissW18}
M.~Weiss and D.~Wichs, ``Is there an oblivious {RAM} lower bound for online reads?'' in \emph{TCC}, 2018, pp. 603--635.

\bibitem{DBLP:conf/crypto/LarsenN18}
K.~G. Larsen and J.~B. Nielsen, ``Yes, there is an oblivious {RAM} lower bound!'' in \emph{CRYPTO}, 2018, pp. 523--542.

\bibitem{DBLP:conf/innovations/BoyleN16}
E.~Boyle and M.~Naor, ``Is there an oblivious {RAM} lower bound?'' in \emph{ITCS}, 2016, pp. 357--368.

\bibitem{DBLP:conf/crypto/KomargodskiL21}
I.~Komargodski and W.~Lin, ``A logarithmic lower bound for oblivious {RAM} (for all parameters),'' in \emph{CRYPTO}, 2021, pp. 579--609.

\bibitem{DBLP:conf/soda/KushilevitzLO12}
E.~Kushilevitz, S.~Lu, and R.~Ostrovsky, ``On the (in)security of hash-based oblivious {RAM} and a new balancing scheme,'' in \emph{SODA}, 2012, pp. 143--156.

\bibitem{DBLP:conf/ccs/StefanovDSFRYD13}
E.~Stefanov, M.~van Dijk, E.~Shi, C.~W. Fletcher, L.~Ren, X.~Yu, and S.~Devadas, ``Path {ORAM:} an extremely simple oblivious {RAM} protocol,'' in \emph{CCS}, 2013, pp. 299--310.

\bibitem{DBLP:conf/focs/PatelP0Y18}
S.~Patel, G.~Persiano, M.~Raykova, and K.~Yeo, ``Panorama: Oblivious {RAM} with logarithmic overhead,'' in \emph{FOCS}, 2018, pp. 871--882.

\bibitem{DBLP:conf/eurocrypt/AsharovKLNPS20}
G.~Asharov, I.~Komargodski, W.~Lin, K.~Nayak, E.~Peserico, and E.~Shi, ``Optorama: Optimal oblivious {RAM},'' in \emph{EUROCRYPT}, 2020, pp. 403--432.

\bibitem{DBLP:conf/crypto/AsharovKLS21}
G.~Asharov, I.~Komargodski, W.~Lin, and E.~Shi, ``Oblivious {RAM} with worst-case logarithmic overhead,'' in \emph{CRYPTO}, 2021, pp. 610--640.

\bibitem{DBLP:conf/soda/ChanCMS19}
T.~H. Chan, K.~Chung, B.~M. Maggs, and E.~Shi, ``Foundations of differentially oblivious algorithms,'' in \emph{SODA}, 2019, pp. 2448--2467.

\bibitem{DBLP:conf/approx/BeimelNZ19}
A.~Beimel, K.~Nissim, and M.~Zaheri, ``Exploring differential obliviousness,'' in \emph{Approximation, Randomization, and Combinatorial Optimization}, 2019, pp. 65:1--65:20.

\bibitem{DBLP:conf/acns/GordonKLX22}
S.~D. Gordon, J.~Katz, M.~Liang, and J.~Xu, ``Spreading the privacy blanket: - differentially oblivious shuffling for differential privacy,'' in \emph{ACNS}, 2022, pp. 501--520.

\bibitem{DBLP:conf/eurocrypt/ZhouSCM23}
M.~Zhou, E.~Shi, T.~H. Chan, and S.~Maimon, ``A theory of composition for differential obliviousness,'' in \emph{EUROCRYPT}, 2023, pp. 3--34.

\bibitem{DBLP:conf/eurocrypt/AsharovCNP0S19}
G.~Asharov, T.~H. Chan, K.~Nayak, R.~Pass, L.~Ren, and E.~Shi, ``Locality-preserving oblivious {RAM},'' in \emph{EUROCRYPT}, 2019, pp. 214--243.

\bibitem{DBLP:conf/ndss/ChakrabortiACMR19}
A.~Chakraborti, A.~J. Aviv, S.~G. Choi, T.~Mayberry, D.~S. Roche, and R.~Sion, ``roram: Efficient range {ORAM} with o(log2 {N)} locality,'' in \emph{NDSS}, 2019.

\bibitem{DBLP:conf/hpec/RenFYDD13}
L.~Ren, C.~W. Fletcher, X.~Yu, M.~van Dijk, and S.~Devadas, ``Integrity verification for path oblivious-ram,'' in \emph{HPEC}, 2013, pp. 1--6.

\bibitem{DBLP:conf/crypto/MathialaganV23}
S.~Mathialagan and N.~Vafa, ``Macorama: Optimal oblivious {RAM} with integrity,'' in \emph{CRYPTO}, 2023, pp. 95--127.

\bibitem{DBLP:conf/pkc/HamlinV21}
A.~Hamlin and M.~Varia, ``Two-server distributed {ORAM} with sublinear computation and constant rounds,'' in \emph{PKC}, 2021, pp. 499--527.

\bibitem{DBLP:conf/ccs/DoernerS17}
J.~Doerner and A.~Shelat, ``Scaling {ORAM} for secure computation,'' in \emph{CCS}, 2017, pp. 523--535.

\bibitem{DBLP:conf/sp/ZahurW0GDEK16}
S.~Zahur, X.~Wang, M.~Raykova, A.~Gasc{\'{o}}n, J.~Doerner, D.~Evans, and J.~Katz, ``Revisiting square-root {ORAM:} efficient random access in multi-party computation,'' in \emph{S\&P}, 2016, pp. 218--234.

\bibitem{DBLP:conf/uss/VadapalliHG23}
A.~Vadapalli, R.~Henry, and I.~Goldberg, ``Duoram: {A} bandwidth-efficient distributed {ORAM} for 2- and 3-party computation,'' in \emph{USENIX Security}, 2023, pp. 3907--3924.

\bibitem{DBLP:conf/ndss/SasyGF18}
S.~Sasy, S.~Gorbunov, and C.~W. Fletcher, ``Zerotrace : Oblivious memory primitives from intel {SGX},'' in \emph{NDSS}, 2018.

\bibitem{DBLP:conf/asiacrypt/ChanGLS17}
T.~H. Chan, Y.~Guo, W.~Lin, and E.~Shi, ``Oblivious hashing revisited, and applications to asymptotically efficient {ORAM} and {OPRAM},'' in \emph{ASIACRYPT}, 2017, pp. 660--690.

\bibitem{DBLP:conf/scn/DittmerO20}
S.~Dittmer and R.~Ostrovsky, ``Oblivious tight compaction in o(n) time with smaller constant,'' in \emph{SCN}, 2020, pp. 253--274.

\bibitem{DBLP:conf/asiacrypt/ChanKNPS18}
T.~H. Chan, J.~Katz, K.~Nayak, A.~Polychroniadou, and E.~Shi, ``More is less: Perfectly secure oblivious algorithms in the multi-server setting,'' in \emph{ASIACRYPT}, 2018, pp. 158--188.

\bibitem{DBLP:conf/acns/JareckiW18}
S.~Jarecki and B.~Wei, ``3pc {ORAM} with low latency, low bandwidth, and fast batch retrieval,'' in \emph{ACNS}, 2018, pp. 360--378.

\bibitem{OS97}
R.~Ostrovsky and V.~Shoup, ``Private information storage (extended abstract),'' in \emph{STOC}, 1997, p. 294–303.

\bibitem{DBLP:conf/icalp/GoodrichM11}
M.~T. Goodrich and M.~Mitzenmacher, ``Privacy-preserving access of outsourced data via oblivious {RAM} simulation,'' in \emph{ICALP}, 2011, pp. 576--587.

\bibitem{DBLP:conf/stoc/AjtaiKS83}
M.~Ajtai, J.~Koml{\'{o}}s, and E.~Szemer{\'{e}}di, ``An o(n log n) sorting network,'' in \emph{STOC}, 1983, pp. 1--9.

\bibitem{DBLP:conf/eurocrypt/BoyleGI15}
E.~Boyle, N.~Gilboa, and Y.~Ishai, ``Function secret sharing,'' in \emph{EUROCRYPT}, 2015, pp. 337--367.

\bibitem{DBLP:conf/ccs/BoyleGI16}
------, ``Function secret sharing: Improvements and extensions,'' in \emph{CCS}, 2016, pp. 1292--1303.

\bibitem{boyle2015oblivious}
E.~Boyle, K.-M. Chung, and R.~Pass, ``Oblivious parallel ram and applications,'' in \emph{TCC}, 2015, pp. 175--204.

\bibitem{DBLP:journals/jal/PaghR04}
R.~Pagh and F.~F. Rodler, ``Cuckoo hashing,'' \emph{J. Algorithms}, vol.~51, no.~2, pp. 122--144, 2004.

\bibitem{cryptoeprint:2021/447}
D.~Noble, ``Explicit, closed-form, general bounds for cuckoo hashing with a stash,'' IACR Cryptol. ePrint Arch., 447, 2021.

\bibitem{DBLP:conf/eurocrypt/GilboaI14}
N.~Gilboa and Y.~Ishai, ``Distributed point functions and their applications,'' in \emph{EUROCRYPT}, 2014, pp. 640--658.

\bibitem{DBLP:conf/eurocrypt/BoyleCG0I0R21}
E.~Boyle, N.~Chandran, N.~Gilboa, D.~Gupta, Y.~Ishai, N.~Kumar, and M.~Rathee, ``Function secret sharing for mixed-mode and fixed-point secure computation,'' in \emph{EUROCRYPT}, 2021, pp. 871--900.

\bibitem{garg2016tworam}
S.~Garg, P.~Mohassel, and C.~Papamanthou, ``Tworam: efficient oblivious ram in two rounds with applications to searchable encryption,'' in \emph{CRYPTO}, 2016, pp. 563--592.

\bibitem{mishra2018oblix}
P.~Mishra, R.~Poddar, J.~Chen, A.~Chiesa, and R.~A. Popa, ``Oblix: An efficient oblivious search index,'' in \emph{S\&P}, 2018, pp. 279--296.

\bibitem{ahmad2019obfuscuro}
A.~Ahmad, B.~Joe, Y.~Xiao, Y.~Zhang, I.~Shin, and B.~Lee, ``Obfuscuro: A commodity obfuscation engine on intel sgx,'' in \emph{NDSS}, 2019.

\bibitem{reichert2024menhir}
L.~Reichert, G.~R. Chandran, P.~Schoppmann, T.~Schneider, and B.~Scheuermann, ``Menhir: An oblivious database with protection against access and volume pattern leakage,'' in \emph{AsiaCCS}, 2024.

\bibitem{DBLP:conf/sp/XuCP15}
Y.~Xu, W.~Cui, and M.~Peinado, ``Controlled-channel attacks: Deterministic side channels for untrusted operating systems,'' in \emph{S\&P}, 2015, pp. 640--656.

\bibitem{DBLP:conf/uss/BulckWKPS17}
J.~V. Bulck, N.~Weichbrodt, R.~Kapitza, F.~Piessens, and R.~Strackx, ``Telling your secrets without page faults: Stealthy page table-based attacks on enclaved execution,'' in \emph{{USENIX} Security}, 2017, pp. 1041--1056.

\bibitem{DBLP:conf/ndss/HoangGY20}
T.~Hoang, J.~Guajardo, and A.~A. Yavuz, ``{MACAO:} {A} maliciously-secure and client-efficient active {ORAM} framework,'' in \emph{NDSS}, 2020.

\bibitem{DBLP:conf/tcc/HemenwayNOSZ23}
B.~Hemenway, D.~Noble, R.~Ostrovsky, M.~Shtepel, and J.~Zhang, ``{DORAM} revisited: Maliciously secure {RAM-MPC} with logarithmic overhead,'' in \emph{TCC}, 2023, pp. 441--470.

\bibitem{DBLP:conf/scn/FalkNO22}
B.~H. Falk, D.~Noble, and R.~Ostrovsky, ``3-party distributed {ORAM} from oblivious set membership,'' in \emph{SCN}, 2022, pp. 437--461.

\bibitem{braun2023ramen}
L.~Braun, M.~Pancholi, R.~Rachuri, and M.~Simkin, ``Ramen: Souper fast three-party computation for ram programs,'' in \emph{CCS}, 2023, pp. 3284--3297.

\bibitem{dan2020DPF}
D.~Boneh, E.~Boyle, H.~Corrigan{-}Gibbs, N.~Gilboa, and Y.~Ishai, ``Lightweight techniques for private heavy hitters,'' \emph{Arxiv: 2012.14884}, 2020.

\bibitem{lam2023gpu}
M.~Lam, J.~Johnson, W.~Xiong, K.~Maeng, U.~Gupta, Y.~Li, L.~Lai, I.~Leontiadis, M.~Rhu, H.-H.~S. Lee \emph{et~al.}, ``Gpu-based private information retrieval for on-device machine learning inference,'' \emph{ArXiv:2301.10904}, 2023.

\end{thebibliography}

\end{document}